\documentclass[11pt]{article} 
\usepackage[margin=1in]{geometry}
\usepackage{amsmath, amssymb, amsthm}
\usepackage{microtype}
\usepackage{lmodern}
\usepackage{hyperref}
\usepackage{cleveref}

\newtheorem{theorem}{Theorem}
\theoremstyle{definition}
\newtheorem{definition}{Definition}
\newtheorem{example}{Example}
\theoremstyle{remark}
\newtheorem{remark}{Remark}

\title{Sketch-Oriented Databases}

\author{ Dominique Duval\thanks{LJK, University Grenoble-Alpes,
    France, dominique.duval@univ-grenoble-alpes.fr} \\ \and Rachid
  Echahed\thanks{LIG, CNRS and University Grenoble-Alpes, France, rachid.echahed@imag.fr} }

\date{}

\usepackage{stmaryrd} 
\usepackage[all]{xy}
\SelectTips{cm}{}
\usepackage{tikz}
\usetikzlibrary{positioning,calc}
\usetikzlibrary{arrows,arrows.meta}
\usetikzlibrary{shapes}
\usetikzlibrary{shapes.misc,shapes.symbols}
\usetikzlibrary{shadows}
\usetikzlibrary{shapes.multipart}
\tikzset{edge/.style = {->,> = latex'}}

\newcommand{\stimes}{\!\times\!} 
 
\newcommand{\scirc}{\!\circ\!} 

\newcommand{\la}{\colon\!\!\colon\!} 
\newcommand{\tyc}{\colon\!} 
\newcommand{\subtype}{\!\sqsubseteq\!} 
\newcommand{\monoto}{\xymatrix@C=1pc{\ar@{>->}[r]&}} 
\newcommand{\isoto}{\xymatrix@C=1pc{\ar[r]^{\simeq}&}}
\newcommand{\coneto}[3]{\underline{#1}:#2\to #3} 
\newcommand{\edto}[1]{\!-\!\!{#1}\!\!\to\!} 
\newcommand{\tr}[3]{{#1}\edto{#2}{#3}} 
 
\newcommand{\tre}[3]{\tr{(#1)}{[#2]}{(#3)}} 
 
\newcommand{\trp}[3]{\tr{(#1)}{#2}{(#3)}} 
\newcommand{\prt}{\shortrightarrow} 
 
\newcommand{\prttt}{\!\prt\!\!} 
\newcommand{\sm}{\underline} 
\newcommand{\cric}{\centerdot} 
\newcommand{\scric}{\!\centerdot\!} 
\newcommand{\id}{\mathit{id}}

\newcommand{\Qu}{\mathbb{Q}} 
\newcommand{\Rk}{\mathbb{R}} 
\newcommand{\Sk}{\mathbb{S}} 
\newcommand{\Tk}{\mathbb{T}} 
\newcommand{\sk}{\mathrm{s}} 
\newcommand{\sketch}[1]{\begin{array}{|c|}\hline #1\hline\end{array}}
\newcommand{\empt}{\mathit{empty}} 
\newcommand{\concat}{\mathit{concat}} 
\newcommand{\Hom}{\mathit{Hom}}
\newcommand{\Mod}{\mathit{Mod}} 
\newcommand{\Prem}{\mathit{Prem}} 
\newcommand{\Conc}{\mathit{Conc}} 
\newcommand{\Th}{\mathit{Th}} 
\newcommand{\Stt}{\mathit{Stt}}  
\newcommand{\Sch}{\mathit{Sch}} 
\newcommand{\Yon}{\mathcal{Y}\!\mathit{on}}
\newcommand{\YPrem}{\mathcal{P}} 
\newcommand{\YConc}{\mathcal{C}} 
\newcommand{\qu}{\mathit{q}} 
\newcommand{\lqu}{{\ell\mathit{q}}} 
\newcommand{\slqu}{{s\ell\mathit{q}}} 
\newcommand{\avp}{\mathit{av}}
\newcommand{\tabl}{\mathit{table}}
\newcommand{\cqu}{\mathit{cq}} 
\newcommand{\clqu}{{\ell\mathit{cq}}}
\newcommand{\obj}{\mathit{set}} 
\newcommand{\Pe}{\mathit{Pe}} 
\newcommand{\Pa}{\mathit{Pa}} 
\newcommand{\Jo}{\mathit{Jo}} 
\newcommand{\Bo}{\mathit{Bo}} 
\newcommand{\au}{\mathit{au}} 
\newcommand{\rep}{\mathit{rep}} 
\newcommand{\String}{\mathit{String}} 
\newcommand{\Cell}{{\it Cell}} 
\newcommand{\NECell}{{\it NECell}} 
\newcommand{\Roww}{{\it Row}} 
\newcommand{\row}{{\it row}} 
\newcommand{\Coll}{{\it Col}} 
\newcommand{\col}{{\it col}} 
\newcommand{\Ent}{{\it Ent}} 
\newcommand{\ent}{{\it ent}} 
\tikzset{MyNode/.style = {rounded rectangle,draw,fill=green!20}}
\tikzset{MyVal/.style = {rounded rectangle,draw,fill=blue!10}}
\tikzset{MyEdge/.style = {rectangle,draw,fill=white}}
\tikzset{MyAtt/.style = {fill=white}} 
\tikzset{MyNodeDouble/.style = {
  rounded corners,rectangle split,rectangle split parts=2,
  rectangle split part fill={green!20,blue!10},draw}}
\tikzset{MyEdgeDouble/.style = {   
  rectangle split,rectangle split parts=2,
  rectangle split part fill={white,blue!10},draw}}
\newcommand{\subsec}[1]{\subsection{#1}} 
\newcommand{\ssubsec}[1]{\subsubsection*{$\bullet$ #1.}} 
\newcommand{\para}[1]{\noindent{}\large\textbf{$\bullet$ #1.}} 
\newcommand{\nota}{\noindent\textbf{Notation. }} 

\begin{document}
\maketitle

\begin{abstract}

\noindent  This paper introduces sketch-oriented databases, a categorical
  framework that encodes database paradigms as finite-limit sketches
  and individual databases and schemas as set-valued models.  It
  illustrates the formalism through graph-oriented paradigms such as
  quivers, RDF triplestores and property graphs. It also 
  shows how common graph features such as labels, attributes, typing, and
  paths, are uniformly captured by sketch constructions. Because paths
  play an important role in queries, we propose inference rules
  formalized via localizers to compute useful paths lazily; such
  localizers are also useful for tasks like database type conformance.
  Finally, the paper introduces stuttering sketches, whose aim is to
  facilitate modular composition and scalable model growth: stuttering
  sketches are finite-limit sketches in which relations are specified
  by a single limit instead of two nested limits, and the paper proves
  that finite unions of models of a stuttering sketch are pointwise
  colimits.

\end{abstract}

\noindent
\textbf{Keywords:} Sketches, Category theory, Databases  

\section{Introduction}
\label{sec:intro} 

The evolution of data management systems has witnessed a significant
shift from rigid, table-centric relational paradigms\footnote{We use
\emph{database paradigm} rather than \emph{database model} to avoid
confusion with other model-related concepts discussed in the paper.}
\cite{Codd1970,Date2003} to more expressive and flexible paradigms
such as graph databases. These systems, ranging from RDF graphs or triplestores
\cite{RDF2014,SPARQL2013} to property graphs
\cite{Angles2018,GQL2024}, are designed to capture complex
relationships and semantic structures inherent in modern data. Yet,
despite their practical success, the different graph database
paradigms still lack a unified formal foundation that supports
rigorous reasoning, compositional semantics, and principled inference.
This paper proposes a categorical framework for various database
paradigms including graph-oriented ones, based on \emph{finite-limit
sketches} \cite{CTCS,TTT}, offering a unified and expressive model-theoretic approach
to database design, semantics and inference.

A graph database typically consists of nodes and edges, where edges
represent binary relationships between nodes. These edges can be
concatenated to form paths, enabling the representation of
combined relationships and navigational queries. 
In categorical
terms, such a structure resembles a small category, where objects
correspond to nodes and morphisms to paths. However, real-world graph
databases incorporate additional features, such as attribute-value pairs,
labels, and typing schemes, that go beyond the basic categorical
structure.

The application of category theory to various facets of database
systems has been explored in the literature since early 1990s;
representative contributions include
\cite{Lellahi1990Categorical,Rosebrugh1991Relational,Tuijn1992Views,Baclawski1994Categorical,Hofstede1996Conceptual,Tuijn1996CGOOD,SpivakW15,WisneskyBJSS17}. These
investigations concentrate mainly on database semantics -- schemas,
databases, and the relationships between them (typing, translation,
and related notions). The relational paradigm has received the
greatest attention, while other paradigms such as Entity-Relationship
(ER), object-oriented databases, and graph-oriented models (for
example RDF) have been explored to a lesser extent.

Building on these categorical foundations, several authors have
employed the concept of sketches \cite{CTCS,TTT}, a diagrammatic
formalism rooted in category theory, to specify diverse aspects of
database paradigms. Notable examples include
\cite{Cadish1996Heterogeneous,johnson_rosebrugh_2002_sketch,ERA,Majkic2015Sketches},
where sketches serve as a flexible language for modeling structural
and semantic properties across different paradigms.

The emergence of multi-paradigm database systems \cite{Lu2018MultiModel}, which integrate multiple databases as models of possibly different paradigms within a single backend, makes categorical tools and sketches especially relevant. Recent work applies categorical constructions to multi-paradigm settings, introducing notions such as schema categories and database categories to link heterogeneous databases \cite{Thiry2018Categories,Liu2019MultiModel,Holubova2021Categorical}. 
These approaches extend earlier work by using categories not only to describe schemas and databases, but also to articulate the relationships among different paradigms.

In contrast to these prior efforts, our approach takes a more
foundational stance by moving one level up. Rather than directly
specifying schemas or databases, we begin by formalizing the paradigms
themselves. Schemas and databases are then treated as models of these
paradigms. 
To this end, we employ the syntactic framework of sketches,
specifically, finite-limit sketches~\cite{CTCS,TTT}, as a
meta-language for specifying core database paradigms such as
relational, RDF triplestore, and property-graph paradigms among others,
thereby grounding their structure and behavior in category theory
\cite{CWM,HCA}.

Finite-limit sketches present categories equipped with chosen finite
limits (products, pullbacks, equalizers, etc.) by means of a quiver (made of nodes and edges) 
enriched with designated loops, triangles and finite cones. A sketch
generates a theory (a category with finite limits) 
and a category of set-valued models (functors preserving those
limits). For example, the sketch for quivers captures the basic
structure of nodes and edges, while extensions of this sketch can
represent labeled quivers, triplestores, and property
graphs. Interpreting a database paradigm as a finite-limit sketch and
an individual database as a model of that sketch brings several
practical advantages: well-behaved limits and colimits in the model
category, canonical adjunctions induced by sketch morphisms, and
direct access to Yoneda lemma for relating syntax and semantics. These
categorical tools capture naturally common database constructions
(typing, attributes, path generation) while making their semantics
explicit and compositional. 
For instance, RDF triplestores
can be modeled as labeled quivers with edges determined by their  
source, target, and label. Property graphs, which include
attribute-value pairs and typing information, are modeled by extending the
previous sketch with additional points and cones representing attributes and
schema constraints. 
A table, as in relational databases, is merely made of cells, which are
(row, column) pairs, and the content of a table can be
seen as a partial function from cells to values. 
These constructions preserve the sketch-oriented
nature of the paradigm, ensuring that all database features are
captured categorically.

Beyond structural modeling, our framework supports inference through 
the notion of \emph{localizer}, a kind of morphism of sketches 
which is an analog for sketches 
of the categorical notion of localization. 
The idea is to distinguish between ``idealized'' databases (with all paths, 
all possible types for each entity, etc.) and more ``concrete'' databases
(with ``lazy'' construction of paths, etc.). 
Each concrete database is a \emph{presentation} of an idealized 
database. The inference process is used for modifying
a presentation without modifying its presented idealized database,
thus, without modifying its semantics. 
A \emph{sketch-oriented inference system} is associated to a localizer and 
its \emph{inference rules} are arrows which are taken to invertible
arrows by the localizer.
The application of inference rules corresponds to pushouts in the category
of concrete databases. 
This approach is an alternative to   
logical inference and supports reasoning over paths, typing
hierarchies, and semantic constraints.

To address scalability and compositionality, we introduce
\emph{stuttering sketches}, a new class of finite-limit sketches
designed to control model size and to support pointwise union of models. 
Stuttering sketches simplify the representation of relations.
This innovation is particularly relevant for managing large graphs. 
We prove that finite unions of 
models of a stuttering sketch are pointwise colimits, ensuring
compositional semantics and tractable inference.

Throughout the paper, we emphasize the categorical foundations and
semantic clarity provided by sketches, advocating for their adoption
as a unifying framework for graph-based data systems and beyond.
 Section~\ref{sec:structure} formalizes sketch-oriented database setting and demonstrates its application to various graph database
 paradigms. Section~\ref{sec:inference} introduces sketch-oriented
 inference systems, defines localizers, and applies them to path
 construction.
Section~\ref{sec:stutter} presents stuttering sketches,
 establishes their theoretical properties, and explores their role in
 model composition. Concluding remarks are provided in
 Section~\ref{sec:conclusion}. We assume that the reader is familiar
 with category theory (see, e.g.,\cite{CWM,HCA} for missing notions). 
 Appendix~\ref{app:tables} has been added to illustrate a
 sketch-based presentation of the relational database paradigm.

\section{Sketch-oriented databases} 
\label{sec:structure}

The aim of this Section is to recall the definitions of finite-limit sketches
and their models, and to provide examples showing that 
a finite-limit sketch can be seen as a database paradigm 
and the models of this sketch as the corresponding databases. 
In Section~\ref{ssec:structure-sketches} some categorical tools are 
reminded, their application to databases is described in 
Section~\ref{ssec:structure-databases} and this is applied 
to several graph database paradigms in Section~\ref{ssec:structure-instances}. 

\subsec{Finite-limit sketches}
\label{ssec:structure-sketches} 

We rely on \cite{CTCS} for an introduction to finite-limit sketches, 
with some minor modifications in the terminology and 
in the way equations are introduced.  
Finite-limit sketches are presentations of categories with finite limits. 
They are defined from simple kinds of graphs called quivers.  
A \emph{quiver} $\Qu$ is made of 
a set of \emph{points} (or \emph{objects}) and a set of \emph{arrows} 
with two functions \emph{domain} and \emph{codomain} from arrows to points. 
An arrow $f$ with domain $X$ and codomain $Y$ is denoted $f:X\to Y$. 
A morphism of quivers is made of 
a function on points and a function on arrows which are compatible with  
domains and codomains. 
In a quiver $\Qu$,  
a \emph{loop} is an arrow $f:X\to X$ and a \emph{triangle} is 
made of three arrows $f:X\to Y$, $g:Y\to Z$ and $h:X\to Z$.
A \emph{diagram} $D$ of \emph{shape} $J$, where $J$ is a quiver, 
is a morphism of quivers from $J$ to $\Qu$. 
A \emph{cone} $C$ 
is made of a diagram $D:J\to\Qu$ called the \emph{base} of $C$, 
a point $V$ of $\Qu$ called the \emph{vertex} of $C$  
and arrows $p_X:V\to D(X)$ of $\Qu$, for all $X$ in $J$, called the  
\emph{projections} of $C$. 
It is a \emph{finite cone} when $J$ is finite. 
At this stage, no commutativity is assumed. 

A \emph{finite-limit sketch} $\Sk$ is made of a quiver $\Qu$,
called the \emph{underlying quiver} of $\Sk$, with 
a set of loops of $\Qu$, called the \emph{potential identities} of $\Sk$, 
a set of triangles of $\Qu$, called the \emph{potential composites} of $\Sk$,  
and a set of finite cones of $\Qu$, called the \emph{potential finite limits} 
of $\Sk$. 
We use the same notations as in categories ($\id_X$ and $g\circ f$) 
for potential identities and potential composites. 
The word ``potential'' here 
means that these features, which have no specific property in the sketch, 
become ``actual'' features in each model (as defined below). 
For instance, a potential identity, which is simply a loop in the sketch, 
becomes an identity function in each model. 
Note that ``potential limits'' may be given other names, 
like ``formal limits'' or ``distinguished cones''.  
For instance, the potential limits below, with their abbreviated notations,
mean respectively 
that the arrow $f$ is a potential monomorphism 
and that $f\equiv g$ is a potential equality, also called an \emph{equation}. 

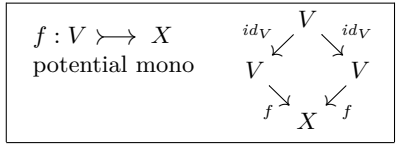
\footnotesize
$$ \begin{array}{|cc|c|cc|}
\cline{1-2}\cline{4-5}
\raisebox{-3ex}{$
\begin{array}{l}
f:V\monoto X \\
\mbox{potential mono} \\ 
\end{array} 
$}
& 
\xymatrix@C=.5pc@R=.6pc{
& V \ar[dl]_{\id_V} \ar[dr]^{\id_V} & \\ 
V \ar[dr]_{f} && V \ar[dl]^{f} \\ 
& X & \\ 
} 
& \null\hspace{1cm}\null &
\raisebox{-3ex}{$
\begin{array}{l}
f\equiv g:V\to X \\
\mbox{equation} \\ 
\end{array} 
$}
& 
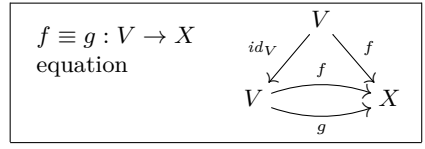
\xymatrix@C=1pc@R=1.5pc{
& V \ar[dl]_{\id_V} \ar[dr]^{f} & \\ 
V \ar@/^/[rr]^{f} \ar@/_/[rr]_{g} && X \\ 
} 
\\
\cline{1-2}\cline{4-5}
\end{array}
$$
\normalsize

On sets, a relation is a subset of a cartesian product. This notion is 
generalized as follows. In a category, a \emph{relation over a diagram $D$}  
is made of two consecutive limits: 
first the limit of $D$, then a mono with codomain the vertex of the limit. 
Note that a relation over a point $X$ is simply a mono with codomain $X$.
A similar definition holds in a sketch, with two consecutive potential
limits. 

A morphism of finite-limit sketches is a morphism of quivers 
which preserves all potential features. 

The \emph{theory} $\Th(\Sk)$ of a finite-limit sketch $\Sk$ 
is the category with finite limits freely generated by $\Sk$ 
where all potential features of $\Sk$ become actual features.  
Every morphism of finite-limit sketches $\sk:\Sk\to\Sk'$ 
generates a  finite-limit preserving functor $\Th(\sk):\Th(\Sk)\to\Th(\Sk')$. 
The morphism $\sk$ is called an \emph{equivalence}  
when the functor $\Th(\sk)$ is an equivalence of categories, 
and the generated equivalence relation is called the 
\emph{equivalence of sketches}.  
For instance, adding a limit over a given diagram is an equivalence of sketches.
A \emph{model} $M$ of a finite-limit sketch $\Sk$ 
(in this paper models are set-valued models, unless explicitly stated) 
associates to each point $X$ of $\Sk$ a set $M(X)$ and 
to each arrow $f:X\to Y$ of $\Sk$ a function $M(f):M(X)\to M(Y)$,
in such a way that each potential feature of $\Sk$ becomes 
an actual feature in the category of sets.
The elements of $M(X)$ are called the \emph{elements of $M$ of sort} $X$. 
A \emph{morphism} $m:M\to M'$ of models of $\Sk$ is a family of 
functions $m_X:M(X)\to M'(X)$ for each point $X$ in $\Sk$, 
such that $m_Y\circ M(f)=M'(f)\circ m_X$ for each arrow $f:X\to Y$ 
in $\Sk$. 
This defines the category $\Mod(\Sk)$ of models of $\Sk$, 
which is equivalent to the category of finite-limit-preserving 
functors from $\Th(\Sk)$ to the category of sets with the natural 
transformations as arrows.
Thus, the categories of models of equivalent finite-limit sketches 
are equivalent categories.

Obviously every morphism of finite-limit sketches $\sk:\Rk\to\Sk$ determines 
a functor $U_\sk:\Mod(\Sk)\to\Mod(\Rk)$ such that 
for each model $M$ of $\Sk$ the model $U_\sk(M)$ of $\Rk$ is defined by 
$U_\sk(M)(-)=M(\sk(-))$. 
An important property of finite-limit sketches is that the functor $U_\sk$
has a left adjoint $F_\sk:\Mod(\Rk)\to\Mod(\Sk)$. 
We call $U_\sk$ and $F_\sk$ the \emph{underlying functor} 
and the \emph{free functor} associated to $\sk$, respectively. 
Another important property of a finite-limit sketch $\Sk$ is that all limits 
and colimits exist in the category $\Mod(\Sk)$.
In addition, all limits in $\Mod(\Sk)$ are pointwise, 
so that they are easily computed from limits of sets. 
Colimits in $\Mod(\Sk)$ are pointwise when $\Sk$ is simply a quiver,
but they are not pointwise in general.
The importance of colimits of models and their computation 
are handled in Sections~\ref{sec:inference} and~\ref{sec:stutter}. 
Note that finite-limit sketches and their categories of models are known 
under several names, corresponding to various points of view: for instance,  
finite-limit sketches correspond to \emph{essentially algebraic 
specifications}  
and their categories of models are the 
\emph{locally finitely presentable categories}.  
A comparison of models of finite-limit sketches with models 
of some kind of generalized Horn theories can be found in \cite{Barr1989}. 
Moreover, there are many variants in the definition of sketches. 
For instance, sketches may be defined either from quivers or from categories.
Most definitions of sketches use diagrams for specifying equalities,
whereas for this purpose we use partial composition and equalizers, 
which are finite limits. 
In this way, the theory and the models of a sketch are defined  
simply by dropping the word ``potential''. 
However, any definition of sketches could be used instead. 

It is often relevant to generalize the notion of morphism of sketches 
into a notion of ``morphism up to equivalence'', based on categorical fractions, 
which we call ``pleomorphism''\footnote{
In biology, pleomorphism is the ability of some micro-organisms to alter 
their shape or size.}. 
A \emph{pleomorphism} $\sk$ from a sketch $\Rk$ to a sketch $\Sk$ 
is made of a sketch $\Sk'$ with two morphisms $\sk_n:\Rk\to\Sk'$ 
and $\sk_d:\Sk\to\Sk'$, repectively called the \emph{numerator} 
and the \emph{denominator} of $\sk$,  
such that $\sk_d$ is an equivalence of sketches. 
Pleomorphisms can be composed, as categorical fractions, 
thanks to pushouts of sketches. 
Every morphism $\sk:\Rk\to\Sk$ can be seen as a pleomorphism with $\sk$ 
as numerator and the identity of $\Sk$ as denominator. 
The notation $\sk:\Rk\to\Sk$ for morphisms is extended to pleomorphisms. 
In fact, like a morphism of sketches, a pleomorphism $\sk:\Rk\to\Sk$
generates a finite-limit-preserving functor $\Th(\sk):\Th(\Rk)\to\Th(\Sk)$ 
and determines an adjunction $F_\sk\dashv U_\sk$ 
between $\Mod(\Sk)$ and $\Mod(\Rk)$. 

In this paper, \emph{sketch} stands for finite-limit sketch   
and the word \emph{potential} may be omitted when there is no ambiguity.

\subsec{Sketch-oriented databases} 
\label{ssec:structure-databases} 

In this paper, we define a \emph{database paradigm} as a sketch $\Sk$ 
and a \emph{database} with respect to this paradigm as a model of $\Sk$. 

It happens that both sketches and graph databases are built from 
\emph{quivers}. 
In order to avoid confusion between the ``meta level'' of
sketches and the ``data level'' of databases, we use two distinct vocabularies
for this unique notion.  
When dealing with databases, a \emph{quiver} 
is made of a set of \emph{nodes}, a set of \emph{edges}, 
and two functions from edges to nodes, respectively called 
\emph{source} and \emph{target}. 
Thus, a quiver is a model of the sketch: 

\footnotesize 
$$ \Sk_\qu = 
\sketch{
\xymatrix@C=2pc@R=1pc{
N && E \ar@/_/[ll]_{s} \ar@/^/[ll]^{t} \\ 
} 
\\ } 
$$ 
\normalsize

In the context of graph databases, several edges may have the same 
source and target. In addition, 
edges are often associated to a third data, typically a string, 
which we call a \emph{label}. 
We define a \emph{labeled quiver} as a quiver where every edge 
is associated to a label.
Thus, a labeled quiver is a model of the sketch:  

\footnotesize
$$ \Sk_{\lqu} = 
\sketch{
\xymatrix@C=2pc@R=1pc{
N && E \ar@/_/[ll]_{s} \ar@/^/[ll]^{t} \ar[r]^{u} & L \\ 
}
\\ } 
$$ 
\normalsize

\nota  In this paper, 
a node $n$ in a quiver is denoted $(n)$, and an edge $e$ with source 
$n_1$ and target $n_2$ is denoted $\tre{n_1}{e}{n_2}$. 
In a labeled quiver, an edge $e$ 
with source $n_1$, target $n_2$ and label $\ell$ is denoted 
$\tre{n_1}{e\la\ell}{n_2}$. 
When there is only one edge from $n_1$ to $n_2$ with label $\ell$ 
this edge is called \emph{strongly labeled} 
and this notation may be simplified as either $\tre{n_1}{\la\ell}{n_2}$ or 
simply $\trp{n_1}{\ell}{n_2}$.

\begin{example}
\label{ex:structure-lq} 
We build a toy database with the following intended meaning:
M.~Barr (MB) is the author of paper \cite{Barr1989} (MHT), 
M.~Barr (MB) and C.~Wells (CW) are the coauthors of the books 
\cite{CTCS} (CTCS) and \cite{TTT} (TTT) which were  
reprinted in Reprints in Theory and Applications of Categories (RTAC), 
and S.~Mac~Lane (SML) is the author of both editions of the book 
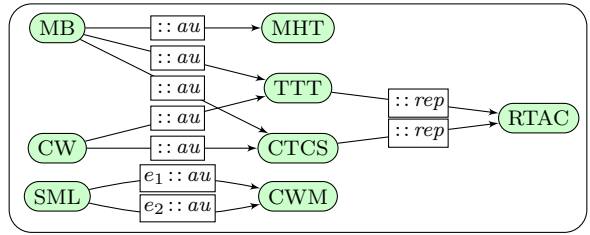
\cite{CWM} (CWM).  

\scriptsize
$$ 
\begin{tikzpicture}[scale=0.8]  
    \node[MyNode] (mb) at (0,2) {MB};
    \node[MyNode] (cw) at (0,0) {CW};
    \node[MyNode] (sml) at (0,-0.8) {SML};
    \node[MyNode] (mht) at (4,2) {MHT};
    \node[MyNode] (ttt) at (4,1) {TTT};
    \node[MyNode] (ctcs) at (4,0) {CTCS};
    \node[MyNode] (cwm) at (4,-0.8) {CWM};
    \node[MyNode] (rtac) at (8,0.5) {RTAC};
    \draw[edge] (mb) to (mht);
    \node[MyEdge] at (2,2) {$\!\la\au\!$};
    \draw[edge] (mb) to (ttt);
    \node[MyEdge] at (2,1.5) {$\!\la\au\!$};
    \draw[edge] (mb) to (ctcs);
    \node[MyEdge] at (2,1) {$\!\la\au\!$};
    \draw[edge] (cw) to (ttt);
    \node[MyEdge] at (2,0.5) {$\!\la\au\!$};
    \draw[edge] (cw) to (ctcs);
    \node[MyEdge] at (2,0) {$\!\la\au\!$};
    \draw[edge][bend left=12] (sml) to (cwm);
    \node[MyEdge] at (2,-0.5) {$\!e_1\la\au\!$};
    \draw[edge][bend right=12] (sml) to (cwm);
    \node[MyEdge] at (2,-1.0) {$\!e_2\la\au\!$};
    \draw[edge] (ttt) to (rtac);
    \node[MyEdge] at (6,0.75) {$\!\la\rep\!$};
    \draw[edge] (ctcs) to (rtac);
    \node[MyEdge] at (6,0.25) {$\!\la\rep\!$};
    \draw[rounded corners=3mm] (-0.8,-1.4) rectangle (8.8,2.4); 
\end{tikzpicture}
$$
\normalsize
\end{example}

A \emph{strongly labeled quiver} is a labeled quiver where 
every edge is strongly labeled, which means that 
every edge is identified by the triple made of its source, target and label. 
This is a kind of \emph{overloading}, since multiple edges 
may share the same label.
Thus, a strongly labeled quiver is a model of the sketch:  

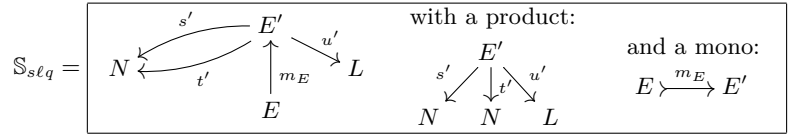
\footnotesize
$$ \Sk_\slqu = 
\sketch{
\raisebox{4ex}{$
\xymatrix@C=1.5pc@R=.3pc{
&& E' \ar@/_/[lld]_{s'} \ar@/^/[lld]^{t'} \ar[rd]^(.6){u'} & \\ 
N && & L \\ 
&& E \ar[uu]_(.4){m_E} &  \\ 
}
$} 
\quad 
\begin{array}{l}
\mbox{with a product:} \\
\xymatrix@C=.7pc@R=1pc{
& E' \ar[dl]_{s'} \ar[d]^{t'} \ar[dr]^{u'} & \\ 
N & N & L \\ 
} 
\\
\end{array} 
\quad 
\begin{array}{l}
\mbox{and a mono:} \\ 
\xymatrix@C=1.5pc{
E\, \ar@{>->}[r]^{m_E} & E' \\ 
} 
\\
\end{array} 
\\ } 
$$ 
\normalsize

\nota  
As above, in a strongly labeled quiver the notation $\tre{n_1}{\la\ell}{n_2}$ 
for edges is simplified as $\trp{n_1}{\ell}{n_2}$.

\begin{example}
\label{ex:structure-slq}
The labeled quiver in Example~\ref{ex:structure-lq} becomes strongly labeled 
when both edges from (SML) to (CWM) are merged. The intended meaning 
of $\trp{n_1}{\au}{n_2}$ is that $n_1$ is an author of $n_2$. 

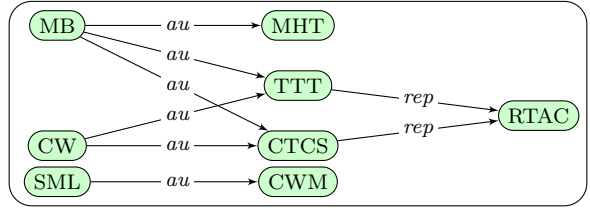
\scriptsize
$$ 
\begin{tikzpicture}[scale=0.8]  
    \node[MyNode] (mb) at (0,2) {MB};
    \node[MyNode] (cw) at (0,0) {CW};
    \node[MyNode] (sml) at (0,-0.6) {SML};
    \node[MyNode] (mht) at (4,2) {MHT};
    \node[MyNode] (ttt) at (4,1) {TTT};
    \node[MyNode] (ctcs) at (4,0) {CTCS};
    \node[MyNode] (cwm) at (4,-0.6) {CWM};
    \node[MyNode] (rtac) at (8,0.5) {RTAC};
    \draw[edge] (mb) to (mht);
    \node[MyAtt] at (2,2) {$\!\au\!$};
    \draw[edge] (mb) to (ttt);
    \node[MyAtt] at (2,1.5) {$\!\au\!$};
    \draw[edge] (mb) to (ctcs);
    \node[MyAtt] at (2,1) {$\!\au\!$};
    \draw[edge] (cw) to (ttt);
    \node[MyAtt] at (2,0.5) {$\!\au\!$};
    \draw[edge] (cw) to (ctcs);
    \node[MyAtt] at (2,0) {$\!\au\!$};
    \draw[edge] (sml) to (cwm);
    \node[MyAtt] at (2,-0.6) {$\!\au\!$};
    \draw[edge] (ttt) to (rtac);
    \node[MyAtt] at (6,0.75) {$\!\rep\!$};
    \draw[edge] (ctcs) to (rtac);
    \node[MyAtt] at (6,0.25) {$\!\rep\!$};
    \draw[rounded corners=3mm] (-0.8,-1.0) rectangle (8.8,2.4); 
\end{tikzpicture}
$$
\normalsize
\end{example}

Obviously every labeled quiver is a quiver
and every strongly labeled quiver is a labeled quiver.
This can be seen as the application of the underlying functors associated 
respectively to the inclusion of $\Sk_\qu$ in $\Sk_\lqu$ 
and to the pleomorphism from $\Sk_\lqu$ to $\Sk_\slqu$ which takes 
$N,E,L,s,t,u$ to $N,E,L,s'\circ m_E,t'\circ m_E,u'\circ m_E$. 

\subsec{Some sketch-oriented database paradigms} 
\label{ssec:structure-instances} 

In this Section we check that several graph database paradigms 
can be seen as sketch-oriented.  Appendix~\ref{app:tables} shows that this is 
also the case for the relational database paradigm.   

\ssubsec{RDF graphs} 
\label{sssec:structure-rdf}

A graph for the \emph{Resource Description Framework}, or \emph{RDF graph}, 
or \emph{RDF triplestore}, is a strongly labeled quiver with some constraints 
on the nature of nodes and edges \cite{RDF2014}. It describes 
\emph{statements} about \emph{resources}. 
A resource corresponds to a node and a statement to an edge, 
which is called an \emph{RDF triple}. 
The label of an RDF triple is called its \emph{predicate}, and its source 
and target are called its \emph{subject} and \emph{object}. 

\ssubsec{ER diagrams} 
\label{sssec:structure-erd}

An \emph{Entity-Relationship diagram}, or \emph{ER diagram}, 
is a strongly labeled quiver which describes interrelated things of interest 
in a specific domain of knowledge. 
It is used for designing databases, either relational or graph-oriented. 
The things of interest are called \emph{entities}, they are classified by 
\emph{entity types}, which correspond to nodes.
An edge between two nodes specifies a \emph{relationship} 
between the entities which are instances of the corresponding entity types. 
The labels of edges are the \emph{relationship types}. 

\begin{example}
\label{ex:structure-erd}
The following ER diagram will be used as a schema for graph databases 
in Example~\ref{ex:structure-schema}. The entity types are 
$\Pe$ for persons, $\Pa$ for papers, $\Bo$ for books and $\Jo$ for journals.
The relationship types are $\au$ for ``is an author of'' and $\rep$ for
``is reprinted in''. 

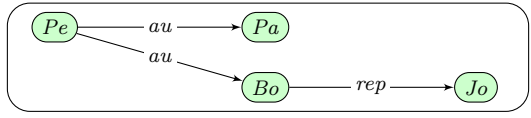
\scriptsize
$$ 
\begin{tikzpicture}[xscale=0.7,yscale=0.8]  
    \node[MyNode] (pe) at (0,2) {$\Pe$};
    \node[MyNode] (pa) at (4,2) {$\Pa$};
    \node[MyNode] (bo) at (4,1) {$\Bo$};
    \node[MyNode] (jo) at (8,1) {$\Jo$};
    \draw[edge] (pe) to (pa);
    \node[MyAtt] at (2,2) {$\!\au\!$};
    \draw[edge] (pe) to (bo);
    \node[MyAtt] at (2,1.5) {$\!\au\!$};
    \draw[edge] (bo) to (jo);
    \node[MyAtt] at (6,1) {$\!\rep\!$};
    \draw[rounded corners=3mm] (-0.9,0.6) rectangle (9.0,2.4); 
\end{tikzpicture}
$$
\normalsize

\end{example}

\ssubsec{Attribute-value pairs} 
\label{sssec:structure-avpairs}

Starting from any kind of sketch-oriented database paradigm $\Sk$  
and any point $X$ in $\Sk$, we check 
that adding a notion of \emph{attribute-value pair} 
(also called \emph{key-value pair} or \emph{property}) on $X$ 
preserves the sketch-oriented nature of the paradigm.
For this purpose, the sketch $\Sk$ is extended with a point $A$ for attributes, 
a point $V$ for values, a point $T$ for item-attribute-value triples,  
and potential limits for asserting that $T$ is a relation on $X$, $A$ and $V$. 
This means that $\Sk$ is merged with the following sketch $\Sk_{\avp,X}$ 
on the point $X$. 

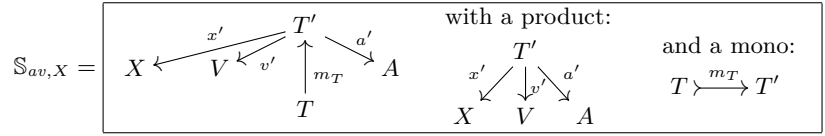
\footnotesize
$$ \Sk_{\avp,X} = 
\sketch{
\raisebox{4ex}{$
\xymatrix@C=1.5pc@R=.3pc{
&& T' \ar[lld]_{x'} \ar[ld]^{\!v'} \ar[rd]^(.6){a'} & \\ 
X & V & & A \\ 
&& T \ar[uu]_(.4){m_T} & \\ 
}$} 
\quad 
\begin{array}{l}
\mbox{with a product:} \\
\xymatrix@C=.7pc@R=1pc{
& T' \ar[dl]_{x'} \ar[d]^{\!v'} \ar[dr]^{a'} & \\ 
X & V & A \\ 
} 
\\
\end{array} 
\quad 
\begin{array}{l}
\mbox{and a mono:} \\ 
\xymatrix@C=1.5pc{
T\, \ar@{>->}[r]^{m_T} & T' \\ 
} 
\\
\end{array} 
\\ } 
$$ 
\normalsize

For instance, by merging the sketches $\Sk_\qu$, $\Sk_{\avp,N}$
and $\Sk_{\avp,E}$ on $N$, $E$, $V$ and $A$, we get a sketch for 
quivers with attribute-value pairs on nodes and edges. 
This is a step towards the construction of a sketch for property graphs
(Definition~\ref{def:property-graph}). 

\nota
We write $a\prttt v$ to denote an attribute-value pair. This notation is 
compatible with the fact that each attribute-value pair $a\prttt v$ on a node $n$ 
of a quiver can be replaced by a strongly labeled edge $\trp{n}{a}{v}$  
if values are seen as nodes. 
Indeed, by identifying $N$ and $V$ in $\Sk_{\avp,N}$ we get a sketch which is 
isomorphic to the sketch $\Sk_\slqu$ for strongly labeled quivers.
 
\begin{example}
\label{ex:structure-avp} 
This example illustrates the transformation of a quiver with attribute-value 
pairs on nodes (on the left) into a quiver with strongly labeled edges
(on the right): 

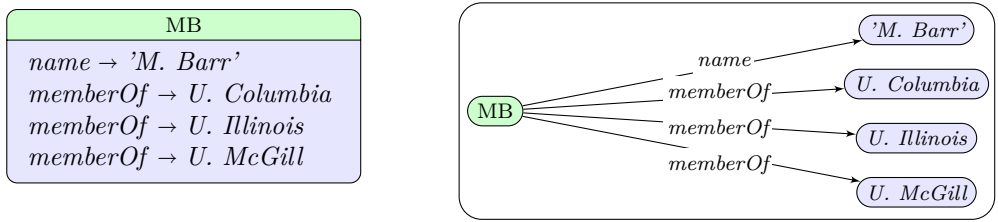
\scriptsize
$$ 
\raisebox{4ex}{
\begin{tikzpicture}[scale=0.8] 
    \node[MyNodeDouble] (mb) at (0,0) {
      MB
      \nodepart{two}
       {\small $ \begin{array}{l} 
        {\it name}\prt \mbox{{\it 'M.~Barr'}} \\ 
        {\it memberOf}\prt \mbox{{\it U. Columbia}} \\ 
        {\it memberOf}\prt \mbox{{\it U. Illinois}} \\ 
        {\it memberOf}\prt \mbox{{\it U. McGill}} \\ 
        \end{array} $ } };
\end{tikzpicture}
}
\qquad\qquad  
\begin{tikzpicture}[xscale=0.8,yscale=1.2] 
    \node[MyNode] (mb) at (0,0.9) {MB}; 
    \node[MyVal] (name) at (7,1.8) {\it 'M.~Barr'}; 
    \node[MyVal] (uc) at (7,1.2) {\it U. Columbia}; 
    \node[MyVal] (ui) at (7,0.6) {\it U. Illinois}; 
    \node[MyVal] (um) at (7,0) {\it U. McGill}; 
    \draw[edge] (mb) to (name);
    \node[MyAtt] at (3.8,1.45) {\!\it name\!};
    \draw[edge] (mb) to (uc);
    \node[MyAtt] at (3.7,1.1) {\!\it memberOf\!};
    \draw[edge] (mb) to (ui);
    \node[MyAtt] at (3.7,0.7) {\!\it memberOf\!};
    \draw[edge] (mb) to (um);
    \node[MyAtt] at (3.7,0.3) {\!\it memberOf\!};
    \draw[rounded corners=3mm] (-0.6,-0.3) rectangle (8.3,2.1); 
\end{tikzpicture}
$$
\normalsize
\end{example}

\ssubsec{Schema (unique typing)} 
\label{sssec:structure-schema}

In contrast to relational databases, 
schemas for graph databases are quite flexible: 
they are not compulsory, they may evolve, 
each node or edge may have several types, 
and there is a hierarchy on types. 
Starting from any kind of sketch-oriented database paradigm, we check 
that adding a notion of schema and typing preserves the sketch-oriented 
nature of the paradigm. Here we introduce a simple notion of schema, 
where each element has exactly one type, then 
(in Section~\ref{sssec:structure-schema})
a general notion, with a hierarchy on types.  

Let us start from a sketch $\Sk$ which specifies some untyped databases. 
In the simplest case, 
both an untyped database $G$ and a schema $\Sch$ are models 
of $\Sk$ and a typing is a morphism $\tau:G\to\Sch$ in $\Mod(\Sk)$. 
However, a schema may have a richer structure, hence the following definitions
for \emph{unique typing}. 
Given a pleomorphism $\sk:\Sk\to\sm{\Sk}$, 
a \emph{schema} for unique typing is a model $\Sch$ of $\sm{\Sk}$ 
and a \emph{uniquely typed database over $\Sch$} is made of a model $G$ of $\Sk$
and a morphism $\tau:G\to U_\sk(\Sch)$ in $\Mod(\Sk)$. 
Thus, with $\sk_n:\Sk\to\sm{\Sk}'$ denoting the numerator of $\sk$,  
the uniquely typed databases are the models of the sketch $\Tk$ 
made of a copy of $\Sk$, a copy of $\sm{\Sk}'$, 
an arrow $\tau_X:X\to \sk_n(X)$ for each point $X$ in $\Sk$ 
and an equation $\tau_Y\circ f \equiv \sk_n(f)\circ\tau_X$ for each arrow 
$f:X\to Y$ in $\Sk$. 

For instance, starting from quivers and deciding that a schema 
is a strongly labeled quiver, we define a \emph{uniquely typed quiver}  
as made of a strongly labeled quiver $\Sch$, a quiver $G$, and a 
morphism (called the \emph{typing morphism}) 
from $G$ to the quiver underlying $\Sch$. 
The \emph{node types} and \emph{edge types} are the nodes and 
edges of $\Sch$. 
Then one can define the \emph{label} of any edge in $G$ as the label 
of its type, this turns $G$ into a labeled quiver, which is not 
strongly labeled in general. 
Thus, uniquely typed quivers are models of the following sketch: 

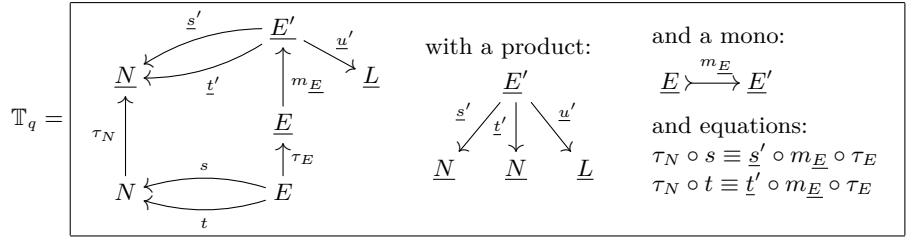
\footnotesize
$$ \Tk_\qu = 
\sketch{
\raisebox{8ex}{$
\xymatrix@C=1.6pc@R=.3pc{
&& \sm{E'} \ar@/_/[lld]_{\sm{s}'} \ar@/^/[lld]^{\sm{t}'} 
  \ar[rd]^(.6){\sm{u}'} & \\ 
\sm{N} && & \sm{L} \\ 
&& \sm{E} \ar[uu]_(.4){m_{\sm{E}}} & \\ 
\\
N \ar[uuu]^{\tau_N} && E \ar@/_/[ll]_{s} \ar@/^/[ll]^{t} \ar[uu]_{\tau_E}  & \\ 
}
$} 
\quad 
\begin{array}{l} 
\mbox{with a product:} \\
\xymatrix@C=1pc@R=1.5pc{
& \sm{E}' \ar[dl]_{\sm{s}'} \ar[d]_{\sm{t}'} \ar[dr]^{\sm{u}'} & \\ 
\sm{N} & \sm{N} & \sm{L} \\ 
} \\
\end{array} 
\quad 
\begin{array}{l}
\mbox{and a mono:} \\ 
\xymatrix@C=1.5pc{
\sm{E}\, \ar@{>->}[r]^{m_{\sm{E}}} & \sm{E}' \\ 
} 
\vspace{1ex}  \\ 
\mbox{and equations: } \\
\tau_N\circ s \equiv \sm{s}' \circ m_{\sm{E}} \circ \tau_E \\
\tau_N\circ t \equiv \sm{t}' \circ m_{\sm{E}} \circ \tau_E \\
\end{array} 
\\ } 
$$ 
\normalsize

\begin{remark}
\label{rem:structure-label} 
This definition clarifies the similarities and differences between 
two notions, here called respectively \emph{edge types} and \emph{labels},  
which occur in various graph databases under various names.
To summarize: 
an edge type is identified by the triple made of its source and target
(which are node types) and its label, 
and the label of an edge is the label of its type. 
Note that typing can be applied to any kind of database paradigm $\Sk$ 
and any point $X$ in $\Sk$, 
whereas labels, as defined in this paper, are restricted to edges
in graph databases.
\end{remark}

\nota The fact that $\sm{x}$ is the type of $x$ is denoted $x\tyc\sm{x}$. 
In addition, each typed edge 
$\tr{(n_1\tyc \sm{n_1})}{[e\tyc\sm{e}]}{(n_2\tyc \sm{n_2})}$ 
may be denoted $\tr{(n_1\tyc \sm{n_1})}{[e\la\ell]}{(n_2\tyc \sm{n_2})}$,
where $\ell$ is the label of $e$, since $\sm{e}$ is recovered 
from this notation as the triple $\trp{\sm{n_1}}{\ell}{\sm{n_2}}$.  

\begin{example}
\label{ex:structure-schema}
The untyped database in Example~\ref{ex:structure-lq} is now augmented 
with types from the ER diagram in Example~\ref{ex:structure-erd}. 

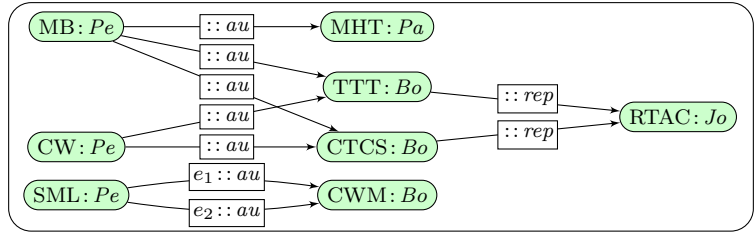
\scriptsize
$$ 
\begin{tikzpicture}[yscale=0.8]  
    \node[MyNode] (mb) at (0,2) {MB$\tyc\Pe$};
    \node[MyNode] (cw) at (0,0) {CW$\tyc\Pe$};
    \node[MyNode] (sml) at (0,-0.8) {SML$\tyc\Pe$};
    \node[MyNode] (mht) at (4,2) {MHT$\tyc\Pa$};
    \node[MyNode] (ttt) at (4,1) {TTT$\tyc\Bo$};
    \node[MyNode] (ctcs) at (4,0) {CTCS$\tyc\Bo$};
    \node[MyNode] (cwm) at (4,-0.8) {CWM$\tyc\Bo$};
    \node[MyNode] (rtac) at (8,0.5) {RTAC$\tyc\Jo$};
    \draw[edge] (mb) to (mht);
    \node[MyEdge] at (2,2) {$\!\la\au\!$};
    \draw[edge] (mb) to (ttt);
    \node[MyEdge] at (2,1.5) {$\!\la\au\!$};
    \draw[edge] (mb) to (ctcs);
    \node[MyEdge] at (2,1) {$\!\la\au\!$};
    \draw[edge] (cw) to (ttt);
    \node[MyEdge] at (2,0.5) {$\!\la\au\!$};
    \draw[edge] (cw) to (ctcs);
    \node[MyEdge] at (2,0) {$\!\la\au\!$};
    \draw[edge][bend left=10] (sml) to (cwm);
    \node[MyEdge] at (2,-0.5) {$\!e_1\la\au\!$};
    \draw[edge][bend right=10] (sml) to (cwm);
    \node[MyEdge] at (2,-1.1) {$\!e_2\la\au\!$};
    \draw[edge] (ttt) to (rtac);
    \node[MyEdge] at (6,0.8) {$\!\la\rep\!$};
    \draw[edge] (ctcs) to (rtac);
    \node[MyEdge] at (6,0.2) {$\!\la\rep\!$};
    \draw[rounded corners=3mm] (-0.9,-1.4) rectangle (9.0,2.4); 
\end{tikzpicture}
$$
\normalsize
\end{example}

Unique typing is easily generalized from quivers to quivers
with attribute-value pairs on nodes and edges, 
where a schema is defined as a strongly labeled quiver with typed attributes 
$a\tyc \sm{a}$ on nodes and edges. 
It follows that uniquely typed quivers 
with attribute-value pairs on nodes and edges are models of a sketch.

\begin{example}
\label{ex:structure-sch-avp} 

Focusing on a single node, this example shows that the notion of
typing with the notation for attribute-value pairs on  nodes (on the
left) corresponds to that of strongly labeled edges (on the right).  The
schema is:

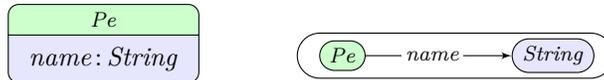
\scriptsize
$$ 
\raisebox{-0.5ex}{
\begin{tikzpicture}
    \node[MyNodeDouble] (mb) at (0,0) {
      $\Pe$
      \nodepart{two}
       {\small $ \begin{array}{l} 
        {\it name}\tyc \String \\ 
        \end{array} $ }}; 
\end{tikzpicture}
}
\qquad\qquad 
\begin{tikzpicture} 
    \node[MyNode] (mb) at (0,0) {$\Pe$}; 
    \node[MyVal] (name) at (2.8,0) {$\String$}; 
    \draw[edge] (mb) to (name);
    \node[MyAtt] at (1.2,0) {$\!\it name\!$};
    \draw[rounded corners=3mm] (-0.6,-0.3) rectangle (3.6,0.3); 
\end{tikzpicture}
$$
\normalsize

\noindent and the uniquely typed quiver with attribute-value pairs on nodes is:

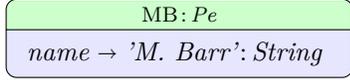
\scriptsize
$$ 
\raisebox{-0.5ex}{
\begin{tikzpicture}
    \node[MyNodeDouble] (mb) at (0,0) {
      MB$\tyc \Pe$
      \nodepart{two}
       {\small $ \begin{array}{l} 
        {\it name}\prt \mbox{{\it 'M.~Barr'}$\tyc\String$} \\ 
        \end{array} $ } };
\end{tikzpicture}
}
\qquad\qquad 
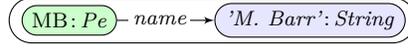
\begin{tikzpicture} 
    \node[MyNode] (mb) at (0,0) {MB$\tyc \Pe$}; 
    \node[MyVal] (name) at (3.2,0) {{\it 'M.~Barr'}$\tyc\String$}; 
    \draw[edge] (mb) to (name);
    \node[MyAtt] at (1.2,0) {\!\it name\!};
    \draw[rounded corners=3mm] (-0.8,-0.3) rectangle (4.6,0.3); 
\end{tikzpicture}
$$
\normalsize
\end{example}

\ssubsec{Property graphs} 
\label{sssec:structure-property}

More generally, given a pleomorphism $\sk:\Sk\to\sm{\Sk}$ 
with a transitive relation on $\sm{\Sk}$ called \emph{subtyping}, 
a \emph{schema} is a model $\Sch$ of $\sm{\Sk}$ 
and a \emph{typed database over $\Sch$} is made of a model $G$ of $\Sk$
and a binary relation (called the \emph{typing relation}) between $G$ and 
$U_\sk(\Sch)$ in $\Mod(\Sk)$, satisfying the \emph{subtyping axiom}: 
if $\sm{x}$ is both a type of $x$ and a subtype of $\sm{x}'$, 
then $\sm{x}'$ is a type of $x$. 
Then the typed databases are the models of the sketch $\Tk$ 
made of a copy of $\Sk$, a copy of $\sm{\Sk}'$ (where $\sk_n:\Sk\to\sm{\Sk}'$ 
is the numerator of $\sk$), 
a relation $\tau$ (the typing relation) between the two,  
and the subtyping axiom. 

\nota The fact that $\sm{x}$ is a type of $x$ is denoted $x\tyc\sm{x}$ 
and the fact that $\sm{x}$ is a subtype of $\sm{x}'$ is denoted 
$\sm{x}\subtype\sm{x}'$. 
Thus, the subtyping axiom states that  
if $x\tyc\sm{x}$ and $\sm{x}\subtype\sm{x}'$ then $x\tyc\sm{x}'$. 

For instance, starting from quivers, this provides  
the sketch-oriented database paradigm of \emph{typed quivers}, 
which can be extended with attribute-value pairs.  
Thus, we propose the following definition of property graphs 
as a sketch-oriented database paradigm.   

\begin{definition}
\label{def:property-graph} 
A \emph{sketch-oriented property graph} is a typed quiver with 
attribute-value pairs on nodes and edges.  
\end{definition}

\section{Sketch-oriented inference}
\label{sec:inference} 

A major interest of all kinds of graph databases is that \emph{paths} 
can be built by concatenating consecutive edges.
Adding \emph{all} paths to a graph often results in a very large graph, 
even infinite, although usually only \emph{some} paths are required.
For instance, the sketch $\Sk_\qu$ for quivers 
can be extended so as to include 
either \emph{total} concatenation for building all paths 
or \emph{partial} concatenation for building some paths, 
and these two sketches are related by a pleomorphism 
(details are given in Section~\ref{ssec:inference-paths}). 
The free functor associated to $\sk_\cqu$ constrains concatenation 
to become total, it can be used for generating paths ``on demand''.  
This is an instance of a general process: 
a \emph{localizer}, as defined below, is a pleomorphism $\sk:\Rk\to\Sk$
which determines a \emph{sketch-oriented inference process}. 
In Section~\ref{ssec:inference-sketches} some categorical tools are 
reminded and localizers are defined, 
their application to inference systems for databases is described in 
Section~\ref{ssec:inference-databases}, and it is applied 
to the construction of paths in Section~\ref{ssec:inference-paths}. 

\subsec{Localizers of sketches}
\label{ssec:inference-sketches}

In category theory, a localization is a functor which adds inverses for 
some arrows. 
Similarly, we say that a morphism of sketches is a localization 
when it adds inverses for some arrows. 
Then we define a \emph{localizer} of sketches as a ``localization up to 
equivalence'', precisely as a pleomorphism where the numerator 
is composed of localizations and equivalence morphisms.  
It follows that the composition of localizers is a localizer. 
Let $\sk:\Rk\to\Sk$ be a localizer and $F_\sk\dashv U_\sk$ 
the associated adjunction, then the free functor $F_\sk$ is a localization. 
Thus, according to a major property of localization in adjunction, 
the underlying functor $U_\sk$ is full and faithful and 
$F_\sk\circ U_\sk$ is an isomorphism. 

For instance, let us check that a morphism $\sk:\Rk\to\Sk$ which takes a cone 
(commutative or not) in $\Rk$ to a limit in $\Sk$ is a localizer. 
First, if $\sk$ adds an equation $f\equiv g$ in $\Sk$ for some parallel arrows 
$f$ and $g$ in $\Rk$, it is a localizer since it makes the equalizer arrow
of $f$ and $g$ invertible. 
Thus, if $\sk$ takes any finite cone of $\Rk$ to a commutative cone, 
it is a localizer. 
Then, if $\sk$ takes a commutative cone of $\Rk$ to a limit, 
it is a localizer since it makes the universal arrow from the vertex of the 
given cone to the vertex of the constructed limit invertible. 

For any sketch $\Sk$, the inclusion of its underlying quiver $\Qu$ 
in $\Sk$ is a localizer. 
Thus, computing the colimit of a diagram $D$ in $\Mod(\Sk)$ 
can be performed in two steps: first compute the colimit of $D$ in $\Mod(\Qu)$, 
which is pointwise, then apply the free functor associated to the inclusion 
of $\Qu$ in $\Sk$. 

The fact that a localizer of sketches $\sk:\Rk\to\Sk$ essentially 
takes some arrows of $\Rk$ to invertible arrows of $\Sk$ 
can be expressed in terms of models of $\Rk$, 
thanks to a ``Yoneda correspondence''. 
Given a sketch $\Rk$, the \emph{contravariant Yoneda model} of $\Rk$ 
is the contravariant model $\Yon_\Rk$ of $\Rk$ in $\Mod(\Rk)$ 
such that $\Yon_\Rk(X)=\Hom_{\Th(\Rk)}(X,-)$ for each point $X$ of $\Rk$ 
and $\Yon_\Rk(f)=-\circ f$ for each arrow $f$ of $\Rk$. 
Then for every pleomorphism $\sk:\Rk\to\Sk$, the free functor $F_\sk$ is such 
that $F_\sk\circ\Yon_\Rk$ and $\Yon_\Sk\circ\sk$ are equivalent. 
%
The \emph{Yoneda Lemma} asserts that 
for each point $X$ of $\Rk$ and model $M$ of $\Rk$ there is a bijection, 
natural in $X$ and $M$, between $M(X)$ and $\Hom_{\Mod(\Rk)}(\Yon_\Rk(X),M)$. 
The morphisms of models from $\Yon_\Rk(X)$ to $M$   
are called the \emph{occurrences} of $X$ in $M$, so that the 
Yoneda Lemma establishes a bijection between the elements of 
$M$ of sort $X$ and the occurrences of $X$ in $M$. 

\subsec{Sketch-oriented inference systems}  
\label{ssec:inference-databases}

We define a \emph{sketch-oriented inference system} 
as a localizer $\sk:\Rk\to\Sk$, and 
a \emph{sketch-oriented inference rule} with respect to $\sk$ 
(with numerator $\sk_n$) as an arrow $r:\Conc\to\Prem$ in $\Rk$ 
such that $\sk_n(r)$ is invertible. 
Then $\Prem$ and $\Conc$ are called 
respectively the \emph{premisse} and the \emph{conclusion} of the rule $r$.  
A model $M$ of $\Rk$ is a \emph{presentation} of the model $F_{\sk}(M)$ 
of $\Sk$. The aim of the inference process 
is to modify presentations without modifying the presented model. 
This process is described below as a colimit in the category $\Mod(\Rk)$. 

Given a rule $r:\Conc\to\Prem$ with respect to an inference system 
$\sk:\Rk\to\Sk$, its image by the Yoneda contravariant model 
is a morphism $\Yon_\Rk(r):\Yon_\Rk(\Prem)\to\Yon_\Rk(\Conc)$ which we call 
the \emph{semantic rule} associated to $r$. 
Then, given a model $M$ of $\Rk$ and an occurrence $p$ of $\Prem$ in $M$, 
the \emph{application} of the rule $r$ to the occurrence $p$  
consists in building the following pushout in $\Mod(\Rk)$
(where $\Yon=\Yon_\Rk$): 

$$ \xymatrix@C=4pc@R=1.5pc{
\Yon(\Prem) \ar[r]^{\Yon(r)} \ar[d]_{p} & 
  \Yon(\Conc) \ar[d]^{q} \\
M \ar[r]^{m} & M_1 \\ 
} $$

Then it can be checked that $M_1$, like $M$, is a presentation of $F_{\sk}(M)$.
First, the morphism $F_{\sk}(\Yon(r))$ is an isomorphism 
since it is equal to $\Yon_\Sk(F_{\sk}(r))$ and $F_{\sk}(r)$ 
is an isomorphism.
It follows that $F_{\sk}(m)$ is also an isomorphism, 
since the image of the pushout by the left adjoint functor $F_{\sk}$ 
is a pushout in $\Mod(\Sk)$. 

The correspondence between the sketch-oriented inference rules and 
the ``usual'' logical rules runs as follows. 
The models of $\Rk$ are the formulas of the logic, 
and the sketch-oriented rule $r$ corresponds to the logical rule 
$\frac{\Yon(\Prem)}{\Yon(\Conc)}$. 
Now, assume that there is a ``good'' notion of union in $\Mod(\Rk)$
(this issue will be studied in Section~\ref{sec:stutter}). 
Given a logical rule $\frac{\YPrem_1\dots\YPrem_n}{\YConc_0}$ 
(with $n\geq0$) let $\YPrem = \YPrem_1\cup\dots\cup\YPrem_n$
and $\YConc=\YPrem\cup\YConc_0$ with the inclusion $\rho:\YPrem\to\YConc$.
Then the corresponding sketch-oriented rule is the arrow $r$ 
in the theory of $\Rk$ such that $\rho=\Yon(r)$. 

For instance, a major application of sketch-oriented inference systems for 
graph databases is the generation of paths, 
see Section~\ref{ssec:inference-paths}. 
Another important application for typed databases with a hierarchy of types 
is the typing process. 

\subsec{Paths}
\label{ssec:inference-paths}

Paths are specific to \emph{graph} databases, they are basic components 
of many algorithms, for instance for pathfinding.  
In this Section we describe the generation of paths on quivers
from a sketch-oriented point of view. This is easily extended to labeled 
quivers, but the extension to strongly labeled quivers requires some care.
This is important since strongly labeled quivers form the structure of 
both RDF graphs (Section~\ref{sssec:structure-rdf}) 
and ER diagrams (Section~\ref{sssec:structure-erd}). 

\para{Paths on quivers}

In the inference system $\sk_\cqu : \Rk_\cqu \to \Sk_\cqu$, 
the sketch $\Rk_\cqu$ is:

\footnotesize 
$$ \Rk_\cqu = 
\sketch{
\raisebox{4ex}{$
\xymatrix@C=3pc@R=1.3pc{
N & 
E \ar@/_/[l]_{s} \ar@/^/[l]^{t} & 
K \ar@/_/[l]_{p_1} \ar@/^/[l]^{p_2} \\ 
N' \ar@/_/[ru]_(.4){\empt} \ar[u]^{m_N} 
&& 
K' \ar@/^/[lu]^(.4){\concat} \ar[u]_{m_K} \\ 
}  $} 
\begin{array}{l}
\mbox{with pullback:} \\ 
\xymatrix@C=.5pc@R=.3pc{
& K \ar[dl]_{p_1} \ar[dr]^{p_2} & \\ 
E \ar[dr]_{t} & & E \ar[dl]^{s} \\ 
& N & 
}  
\end{array}
\begin{array}{l}
\mbox{and monos:} \\ 
\xymatrix@C=1.5pc@R=0pc{
N'\, \ar@{>->}[r]^{m_N} & N \\ 
K'\, \ar@{>->}[r]^{m_K} & K \\ 
} 
\end{array}
\begin{array}{l}
\mbox{and equations:} \\ 
s\scirc\concat \equiv s\scirc p_1 \scirc m_K \\
t\scirc\concat \equiv t\scirc p_2 \scirc m_K \\
s\scirc\empt \equiv m_N \\
t\scirc\empt \equiv m_N \\
\end{array}
\\ } 
$$ 
\normalsize 

By identifying $N'$ and $N$, as well as $K'$ and $K$, taking 
$m_N$ and $m_K$ to $\id_N$ and $\id_K$ respectively, 
and adding equations for ensuring that $\concat$ is associative
with $\empt$ as unit, we get a sketch $\Sk_\cqu$ for quivers with all paths  
(which means, a sketch for small categories) together with 
a pleomorphism $\sk_\cqu:\Rk_\cqu\to\Sk_\cqu$ which is a localizer.

\nota 
For short, concatenation is also denoted 
$e_1 \scric e_2 = \concat(e_1,e_2)$. 

\begin{example}
\label{ex:inference-cq}
The semantic rule associated to the rule $m_K:K'\to K$ is: 

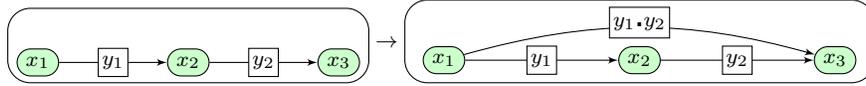
\scriptsize
$$ 
\begin{tikzpicture}[yscale=0.9,xscale=1]  
    \node[MyNode] (n1) at (0,0) {$x_1$};
    \node[MyNode] (n2) at (2,0) {$x_2$};
    \node[MyNode] (n3) at (4,0) {$x_3$};
    \draw[edge] (n1) to (n2);
    \node[MyEdge] at (1,0) {$\!y_1\!$};
    \draw[edge] (n2) to (n3);
    \node[MyEdge] at (3,0) {$\!y_2\!$};
    \draw[rounded corners=3mm] (-0.4,-0.3) rectangle (4.4,0.8); 
\end{tikzpicture}
\raisebox{4ex}{$\;\to\;$} 
\begin{tikzpicture}[yscale=1,xscale=1.3] 
    \node[MyNode] (n1) at (0,0) {$x_1$};
    \node[MyNode] (n2) at (2,0) {$x_2$};
    \node[MyNode] (n3) at (4,0) {$x_3$};
    \draw[edge] (n1) to (n2);
    \node[MyEdge] at (1,0) {$\!y_1\!$};
    \draw[edge] (n2) to (n3);
    \node[MyEdge] at (3,0) {$\!y_2\!$};
    \draw[edge][bend left=20] (n1) to (n3);
    \node[MyEdge] at (2,0.5) {$\!y_1\scric y_2\!$};
    \draw[rounded corners=3mm] (-0.4,-0.3) rectangle (4.4,0.8); 
\end{tikzpicture}
$$
\normalsize
There are two distinct occurrences of the premisse of this rule in 
the quiver below on the left. The application of the rule 
successively to both occurrences returns the quiver on the right: 

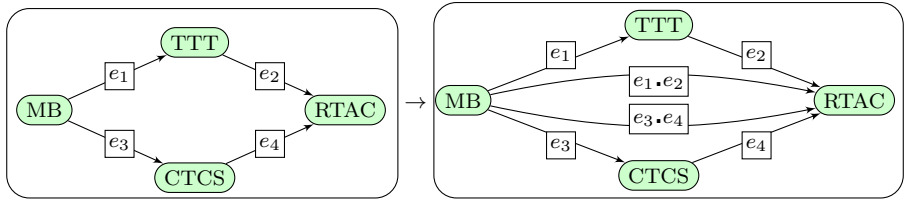
\scriptsize
$$ 
\begin{tikzpicture}[yscale=0.9,xscale=1]  
    \node[MyNode] (mb) at (0,1) {MB};
    \node[MyNode] (ttt) at (2,2) {TTT};
    \node[MyNode] (ctcs) at (2,0) {CTCS};
    \node[MyNode] (rtac) at (4,1) {RTAC};
    \draw[edge] (mb) to (ttt);
    \node[MyEdge] at (1,1.5) {$\!e_1\!$};
    \draw[edge] (mb) to (ctcs);
    \node[MyEdge] at (1,0.5) {$\!e_3\!$};
    \draw[edge] (ttt) to (rtac);
    \node[MyEdge] at (3,1.5) {$\!e_2\!$};
    \draw[edge] (ctcs) to (rtac);
    \node[MyEdge] at (3,0.5) {$\!e_4\!$};
    \draw[rounded corners=3mm] (-0.5,-0.3) rectangle (4.7,2.5); 
\end{tikzpicture}
\raisebox{10ex}{$\;\to\;$} 
\begin{tikzpicture}[yscale=1,xscale=1.3] 
    \node[MyNode] (mb) at (0,1) {MB};
    \node[MyNode] (ttt) at (2,2) {TTT};
    \node[MyNode] (ctcs) at (2,0) {CTCS};
    \node[MyNode] (rtac) at (4,1) {RTAC};
    \draw[edge] (mb) to (ttt);
    \node[MyEdge] at (1,1.6) {$\!e_1\!$};
    \draw[edge] (mb) to (ctcs);
    \node[MyEdge] at (1,0.4) {$\!e_3\!$};
    \draw[edge] (ttt) to (rtac);
    \node[MyEdge] at (3,1.6) {$\!e_2\!$};
    \draw[edge] (ctcs) to (rtac);
    \node[MyEdge] at (3,0.4) {$\!e_4\!$};
    \draw[edge][bend left=15] (mb) to (rtac);
    \node[MyEdge] at (2,1.25) {$\!e_1\scric e_2\!$};
    \draw[edge][bend right=15]  (mb) to (rtac);
    \node[MyEdge] at (2,0.75) {$\!e_3 \scric e_4\!$};
    \draw[rounded corners=3mm] (-0.3,-0.3) rectangle (4.5,2.3); 
\end{tikzpicture}
$$
\normalsize
\end{example}

\para{Paths on labeled quivers}  

Paths on quivers are easily generalized to paths on 
labeled quivers, assuming that all labels may be concatenated.  
Thus, the localizer $\sk_\cqu$ is extended as a localizer 
$\sk_\clqu:\Rk_\clqu\to\Sk_\clqu$ 
with compatible concatenation on edges and on labels. 

\nota 
The notation for concatenation of edges is easily extended to labeled edges: 
$e_1 \scric e_2\la \ell_1\cric\ell_2
= \concat(e_1\la \ell_1,e_2\la \ell_2)$. 

\begin{example}
\label{ex:inference-clq} 
Example~\ref{ex:inference-cq} is now extended to labeled quivers.  
The semantic rule associated to $m_K:K'\to K$ is: 

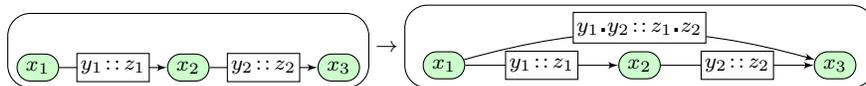
\scriptsize
$$ 
\begin{tikzpicture}[yscale=0.9,xscale=1]  
    \node[MyNode] (n1) at (0,0) {$x_1$};
    \node[MyNode] (n2) at (2,0) {$x_2$};
    \node[MyNode] (n3) at (4,0) {$x_3$};
    \draw[edge] (n1) to (n2);
    \node[MyEdge] at (1,0) {$\!y_1\la z_1\!$};
    \draw[edge] (n2) to (n3);
    \node[MyEdge] at (3,0) {$\!y_2\la z_2\!$};
    \draw[rounded corners=3mm] (-0.4,-0.3) rectangle (4.4,0.8); 
\end{tikzpicture}
\raisebox{4ex}{$\;\to\;$} 
\begin{tikzpicture}[yscale=1,xscale=1.3] 
    \node[MyNode] (n1) at (0,0) {$x_1$};
    \node[MyNode] (n2) at (2,0) {$x_2$};
    \node[MyNode] (n3) at (4,0) {$x_3$};
    \draw[edge] (n1) to (n2);
    \node[MyEdge] at (1,0) {$\!y_1 \la z_1\!$};
    \draw[edge] (n2) to (n3);
    \node[MyEdge] at (3,0) {$\!y_2 \la z_2\!$};
    \draw[edge][bend left=20] (n1) to (n3);
    \node[MyEdge] at (2,0.5) {$\!y_1\scric y_2 \la z_1\scric z_2\!$};
    \draw[rounded corners=3mm] (-0.4,-0.3) rectangle (4.4,0.8); 
\end{tikzpicture}
$$
\normalsize
When this rule is applied to both occurrences of its premisse we get: 

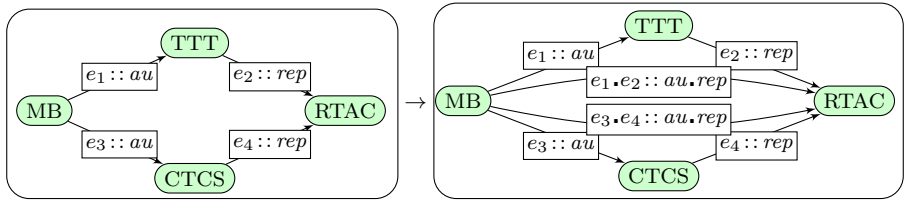
\scriptsize
$$ 
\begin{tikzpicture}[yscale=0.9,xscale=1]  
    \node[MyNode] (mb) at (0,1) {MB};
    \node[MyNode] (ttt) at (2,2) {TTT};
    \node[MyNode] (ctcs) at (2,0) {CTCS};
    \node[MyNode] (rtac) at (4,1) {RTAC};
    \draw[edge] (mb) to (ttt);
    \node[MyEdge] at (1,1.5) {$\!e_1\la\au\!$};
    \draw[edge] (mb) to (ctcs);
    \node[MyEdge] at (1,0.5) {$\!e_3\la\au\!$};
    \draw[edge] (ttt) to (rtac);
    \node[MyEdge] at (3,1.5) {$\!e_2\la\rep\!$};
    \draw[edge] (ctcs) to (rtac);
    \node[MyEdge] at (3,0.5) {$\!e_4\la\rep\!$};
    \draw[rounded corners=3mm] (-0.5,-0.3) rectangle (4.7,2.5); 
\end{tikzpicture}
\raisebox{10ex}{$\;\to\;$} 
\begin{tikzpicture}[yscale=1,xscale=1.3] 
    \node[MyNode] (mb) at (0,1) {MB};
    \node[MyNode] (ttt) at (2,2) {TTT};
    \node[MyNode] (ctcs) at (2,0) {CTCS};
    \node[MyNode] (rtac) at (4,1) {RTAC};
    \draw[edge] (mb) to (ttt);
    \node[MyEdge] at (1,1.6) {$\!e_1\la\au\!$};
    \draw[edge] (mb) to (ctcs);
    \node[MyEdge] at (1,0.4) {$\!e_3\la\au\!$};
    \draw[edge] (ttt) to (rtac);
    \node[MyEdge] at (3,1.6) {$\!e_2\la\rep\!$};
    \draw[edge] (ctcs) to (rtac);
    \node[MyEdge] at (3,0.4) {$\!e_4\la\rep\!$};
    \draw[edge][bend left=15] (mb) to (rtac);
    \node[MyEdge] at (2,1.25) {$\!e_1\scric e_2 \la \au\scric\rep\!$};
    \draw[edge][bend right=15]  (mb) to (rtac);
    \node[MyEdge] at (2,0.75) {$\!e_3 \scric e_4 \la \au\scric\rep\!$};
    \draw[rounded corners=3mm] (-0.3,-0.3) rectangle (4.5,2.3); 
\end{tikzpicture}
$$
\normalsize
\end{example}

\para{Paths on strongly labeled quivers}

There are two different ways to define concatenation of edges in a 
strongly labeled quiver, corresponding to two different intended meanings  
of paths. 
In the first definition, concatenation of labels is a binary operation 
on labels, as for labeled quivers, with notation $ \ell_1\cric\ell_2
= \concat(\ell_1,\ell_2)$. 
By associativity of concatenation we get paths of the form 
$\trp{n_1}{\ell_1\cric\ell_2\cric\dots\cric\ell_k}{n_{k\!+\!1}}$. 
In the second definition, concatenation of labels takes two labels 
$\ell_1$, $\ell_2$ and a node $n$ to a label denoted $\ell_1 (n) \ell_2$. 
By associativity of concatenation we get paths of the form 
$\trp{n_1}{\ell_1(n_2)\ell_2\dots(n_k)\ell_k}{n_{k\!+\!1}}$. 

\begin{example}
\label{ex:inference-cslq} 
With the first definition of concatenation of labels, here are the semantic rule 
and its application successively to both occurrences of the premisse of 
the rule, which both yield the same result. 

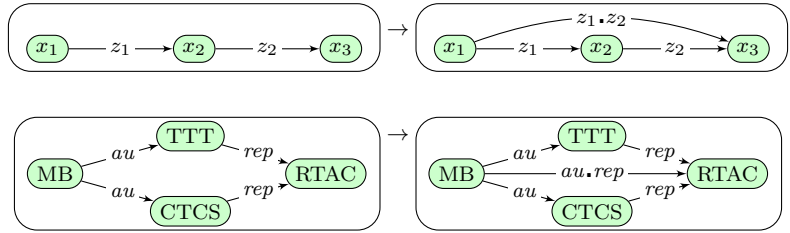
\scriptsize
$$ 
\begin{tikzpicture}[yscale=1,xscale=1.3] 
    \node[MyNode] (n1) at (0,0) {$x_1$};
    \node[MyNode] (n2) at (1.5,0) {$x_2$};
    \node[MyNode] (n3) at (3,0) {$x_3$};
    \draw[edge] (n1) to (n2);
    \node[MyAtt] at (0.75,0) {$\!z_1\!$};
    \draw[edge] (n2) to (n3);
    \node[MyAtt] at (2.25,0) {$\!z_2\!$};
    \draw[rounded corners=3mm] (-0.4,-0.3) rectangle (3.4,0.6); 
\end{tikzpicture}
\raisebox{4ex}{$\;\to\;$} 
\begin{tikzpicture}[yscale=1,xscale=1.3] 
    \node[MyNode] (n1) at (0,0) {$x_1$};
    \node[MyNode] (n2) at (1.5,0) {$x_2$};
    \node[MyNode] (n3) at (3,0) {$x_3$};
    \draw[edge] (n1) to (n2);
    \node[MyAtt] at (0.75,0) {$\!z_1\!$};
    \draw[edge] (n2) to (n3);
    \node[MyAtt] at (2.25,0) {$\!z_2\!$};
    \draw[edge][bend left=25] (n1) to (n3);
    \node[MyAtt] at (1.5,0.4) {$\!z_1\scric z_2\!$};
    \draw[rounded corners=3mm] (-0.4,-0.3) rectangle (3.4,0.6); 
\end{tikzpicture}
$$
\normalsize

\scriptsize
$$ 
\begin{tikzpicture}[yscale=0.5,xscale=0.9]  
    \node[MyNode] (mb) at (0,1) {MB};
    \node[MyNode] (ttt) at (2,2) {TTT};
    \node[MyNode] (ctcs) at (2,0) {CTCS};
    \node[MyNode] (rtac) at (4,1) {RTAC};
    \draw[edge] (mb) to (ttt);
    \node[MyAtt] at (1,1.5) {$\!\au\!$};
    \draw[edge] (mb) to (ctcs);
    \node[MyAtt] at (1,0.5) {$\!\au\!$};
    \draw[edge] (ttt) to (rtac);
    \node[MyAtt] at (3,1.5) {$\!\rep\!$};
    \draw[edge] (ctcs) to (rtac);
    \node[MyAtt] at (3,0.5) {$\!\rep\!$};
    \draw[rounded corners=3mm] (-0.6,-0.5) rectangle (4.8,2.5); 
\end{tikzpicture}
\raisebox{10ex}{$\;\to\;$} 
\begin{tikzpicture}[yscale=0.5,xscale=0.9] 
    \node[MyNode] (mb) at (0,1) {MB};
    \node[MyNode] (ttt) at (2,2) {TTT};
    \node[MyNode] (ctcs) at (2,0) {CTCS};
    \node[MyNode] (rtac) at (4,1) {RTAC};
    \draw[edge] (mb) to (ttt);
    \node[MyAtt] at (1,1.5) {$\!\au\!$};
    \draw[edge] (mb) to (ctcs);
    \node[MyAtt] at (1,0.5) {$\!\au\!$};
    \draw[edge] (ttt) to (rtac);
    \node[MyAtt] at (3,1.5) {$\!\rep\!$};
    \draw[edge] (ctcs) to (rtac);
    \node[MyAtt] at (3,0.5) {$\!\rep\!$};
    \draw[edge] (mb) to (rtac);
    \node[MyAtt] at (2,1) {$\!\au\scric\rep\!$};
    \draw[rounded corners=3mm] (-0.6,-0.5) rectangle (4.8,2.5); 
\end{tikzpicture}
$$
\normalsize

\noindent With the second definition of concatenation of labels, 
here are the semantic rule and its application successively to both 
occurrences of the premisse, 
which yield two paths. 

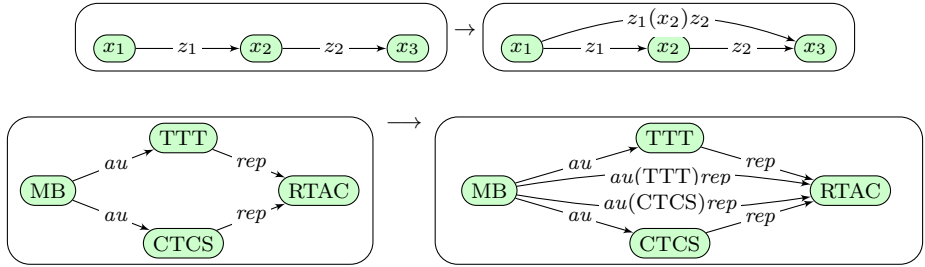
\scriptsize
$$ 
\begin{tikzpicture}[yscale=1,xscale=1.3] 
    \node[MyNode] (n1) at (0,0) {$x_1$};
    \node[MyNode] (n2) at (1.5,0) {$x_2$};
    \node[MyNode] (n3) at (3,0) {$x_3$};
    \draw[edge] (n1) to (n2);
    \node[MyAtt] at (0.75,0) {$\!z_1\!$};
    \draw[edge] (n2) to (n3);
    \node[MyAtt] at (2.25,0) {$\!z_2\!$};
    \draw[rounded corners=3mm] (-0.4,-0.3) rectangle (3.4,0.6); 
\end{tikzpicture}
\raisebox{4ex}{$\;\to\;$} 
\begin{tikzpicture}[yscale=1,xscale=1.3] 
    \node[MyNode] (n1) at (0,0) {$x_1$};
    \node[MyNode] (n2) at (1.5,0) {$x_2$};
    \node[MyNode] (n3) at (3,0) {$x_3$};
    \draw[edge] (n1) to (n2);
    \node[MyAtt] at (0.75,0) {$\!z_1\!$};
    \draw[edge] (n2) to (n3);
    \node[MyAtt] at (2.25,0) {$\!z_2\!$};
    \draw[edge][bend left=25] (n1) to (n3);
    \node[MyAtt] at (1.5,0.4) {$\!z_1(x_2)z_2\!$};
    \draw[rounded corners=3mm] (-0.4,-0.3) rectangle (3.4,0.6); 
\end{tikzpicture}
$$
\normalsize

\scriptsize
$$ 
\begin{tikzpicture}[xscale=0.9,yscale=0.7]  
    \node[MyNode] (mb) at (0,1) {MB};
    \node[MyNode] (ttt) at (2,2) {TTT};
    \node[MyNode] (ctcs) at (2,0) {CTCS};
    \node[MyNode] (rtac) at (4,1) {RTAC};
    \draw[edge] (mb) to (ttt);
    \node[MyAtt] at (1,1.5) {$\!\au\!$};
    \draw[edge] (mb) to (ctcs);
    \node[MyAtt] at (1,0.5) {$\!\au\!$};
    \draw[edge] (ttt) to (rtac);
    \node[MyAtt] at (3,1.5) {$\!\rep\!$};
    \draw[edge] (ctcs) to (rtac);
    \node[MyAtt] at (3,0.5) {$\!\rep\!$};
    \draw[rounded corners=3mm] (-0.6,-0.4) rectangle (4.8,2.4); 
\end{tikzpicture}
\raisebox{15ex}{$\;\;\longrightarrow\;\;$} 
\begin{tikzpicture}[xscale=1.2,yscale=0.7] 
    \node[MyNode] (mb) at (0,1) {MB};
    \node[MyNode] (ttt) at (2,2) {TTT};
    \node[MyNode] (ctcs) at (2,0) {CTCS};
    \node[MyNode] (rtac) at (4,1) {RTAC};
    \draw[edge] (mb) to (ttt);
    \node[MyAtt] at (1,1.5) {$\!\au\!$};
    \draw[edge] (mb) to (ctcs);
    \node[MyAtt] at (1,0.5) {$\!\au\!$};
    \draw[edge] (ttt) to (rtac);
    \node[MyAtt] at (3,1.5) {$\!\rep\!$};
    \draw[edge] (ctcs) to (rtac);
    \node[MyAtt] at (3,0.45) {$\!\rep\!$};
    \draw[edge][bend left=15] (mb) to (rtac);
    \node[MyAtt] at (2,1.25) {$\!\au({\rm TTT})\rep\!$};
    \draw[edge][bend right=15]  (mb) to (rtac);
    \node[MyAtt] at (2,0.75) {$\!\au({\rm CTCS})\rep\!$};
    \draw[rounded corners=3mm] (-0.6,-0.4) rectangle (4.8,2.4); 
\end{tikzpicture}
$$
\normalsize

\end{example}

\para{Paths on property graphs} 

In Section~\ref{sssec:structure-schema} we have defined a sketch-oriented 
property graph as a typed quiver with attribute-value pairs on nodes and edges.  
Then types form a strongly labeled quiver (the schema), 
so that we have to choose one of both definitions above 
for the concatenation of labels. We propose to choose the second definition, 
which seems to be the most usual one, so that paths have the form 
$(n_1)\ell_1(n_2)\ell_2\dots(n_k)\ell_k(n_{k\!+\!1})$.

\section{Stuttering sketches} 
\label{sec:stutter}

The notion of \emph{relation over a diagram}, as defined in 
Section~\ref{ssec:structure-sketches}, is used in many sketches related to 
databases in Sections~\ref{sec:structure} and~\ref{sec:inference}. 
It is built from two nested limits:
a limit over a diagram followed by a monomorphism. 
In set theory, a relation is defined as a subset of a cartesian product, 
and it is well-known that the cartesian product may be dropped by 
defining a relation as a \emph{jointly monic} family of functions. 
This Section is devoted to a similar issue: we define \emph{stuttering
sketches} as finite-limit sketches where relations over a diagram $D$ 
are specified by only one limit cone, thus avoiding the use of the 
limit of $D$. A remarkable property is that 
the union of models of a stuttering sketch is a pointwise colimit. 
As far as we know, this provides a new notion and a new result 
in sketch theory. 
In Section~\ref{ssec:stutter-colimits} we define pointwise limits
and colimits of models of a sketch. In this paper all sketches 
are finite-limit sketches, so that all limits of models are pointwise. 
But it is not the case for colimits, and we check that the fact of 
being pointwise for a colimit is not stable under equivalence of sketches. 
In Section~\ref{ssec:stutter-sketches} we define \emph{stuttering sketches}   
as specific finite-limit sketches where 
relations are specified by a unique limit instead of two nested limits 
(Theorem~\ref{theo:stutter}). 
Then we prove that finite unions of models 
of a stuttering sketch are pointwise colimits (Theorem~\ref{theo:union}). 

\subsec{Pointwise colimits of models} 
\label{ssec:stutter-colimits}

Pointwise limits and colimits in the category of models of a sketch 
are defined as follows. 
Let $\Sk_\obj$ be the sketch made of a unique point $X_\obj$ and no arrow,
so that $\Mod(\Sk_\obj)$ is the category of sets. 
For every sketch $\Sk$ and every point $X$ in $\Sk$, 
let $\sk_X : \Sk_\obj \to \Sk$ be the morphism which takes $X_\obj$ to $X$. 
Then a limit or colimit of models of $\Sk$ is \emph{pointwise} 
if it is preserved by the underlying functor associated to $\sk_X$ 
for each point $X$ of $\Sk$. 
Note that this property depends on the points of $\Sk$, 
so that it is not stable under equivalence of sketches:
see Examples~\ref{ex:stutter-not-pointwise} and~\ref{ex:stutter-pointwise}. 

Limits and colimits in the category of models of a sketch can be used 
for defining analogs of the set-theoretic notions of 
inclusion, intersection and union.
The \emph{inclusion} of models of $\Sk$ is defined pointwise: 
a model $M'$ is included in a model $M$, denoted $M'\subseteq M$, if 
there is a morphism $m:M'\to M$ such that the function $m_X:M'(X)\to M(X)$ 
is an inclusion for every point $X$ of $\Sk$.
Then $m$ is a monomorphism in $\Mod(\Sk)$.
Given some models of $\Sk$, we say that they are \emph{compatible} 
when they are included in a common model $M_\ast$ of $\Sk$. 
Let $M_1$ and $M_2$ be two compatible models of $\Sk$. 
Then the \emph{intersection} $M_1\cap M_2$ is 
the vertex of the pullback of the inclusions 
$M_1\subseteq M_\ast \supseteq M_2$. 
Since it is a limit, it is pointwise.  
Then the \emph{union} $M_1\cup M_2$ is defined
as the pushout of the inclusions $M_1\supseteq M_1\cap M_2 \subseteq M_2$ 
in 
$\Mod(\Sk)$. The pushout exists but it is not pointwise in general.

\begin{example}
\label{ex:stutter-not-pointwise} 
We consider the inference system $\sk_\cqu : \Rk_\cqu \to \Sk_\cqu$ 
for partial-to-total concatenation, as in Section~\ref{ssec:inference-paths}. 
The following commutative squares illustrate the computation of unions,  
first in $\Mod(\Sk_\cqu)$ (on the left), then in $\Mod(\Rk_\cqu)$
(on the right). For readability, empty paths are omitted. 

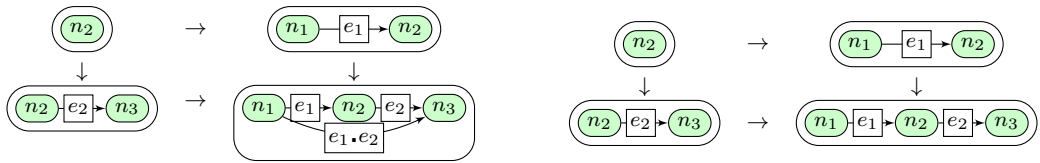
\scriptsize
$$ 
\begin{array}{ccc}
\begin{tikzpicture}
    \node[MyNode] (n2) at (1.5,0) {$n_2$};
    \draw[rounded corners=3mm] (1.1,-0.3) rectangle (1.9,0.3); 
\end{tikzpicture}
& \raisebox{2ex}{$\to$} & 
    \begin{tikzpicture}
    \node[MyNode] (n1) at (0,0) {$n_1$};
    \node[MyNode] (n2) at (1.5,0) {$n_2$};
    \draw[edge] (n1) to (n2);
    \node[MyEdge] at (0.75,0) {$\!e_1\!$};
    \draw[rounded corners=3mm] (-0.4,-0.3) rectangle (1.9,0.3); 
\end{tikzpicture} \\
\downarrow & & \downarrow \\ 
\begin{tikzpicture}
    \node[MyNode] (n2) at (1.5,0) {$n_2$};
    \node[MyNode] (n3) at (2.7,0) {$n_3$};
    \draw[edge] (n2) to (n3);
    \node[MyEdge] at (2.05,0) {$\!e_2\!$};
    \draw[rounded corners=3mm] (1.1,-0.3) rectangle (3.1,0.3); 
    \node (void) at (2,-0.6) {};
\end{tikzpicture} 
& \raisebox{6ex}{$\to$} & 
\begin{tikzpicture}
    \node[MyNode] (n1) at (0,0) {$n_1$};
    \node[MyNode] (n2) at (1.2,0) {$n_2$};
    \node[MyNode] (n3) at (2.4,0) {$n_3$};
    \draw[edge] (n1) to (n2);
    \node[MyEdge] at (0.55,0) {$\!e_1\!$};
    \draw[edge] (n2) to (n3);
    \node[MyEdge] at (1.75,0) {$\!e_2\!$};
    \draw[edge][bend right=25] (n1) to (n3);
    \node[MyEdge] at (1.2,-0.4) {$\!e_1\scric e_2\!$};
    \draw[rounded corners=3mm] (-0.4,-0.7) rectangle (2.8,0.3); 
\end{tikzpicture} \\
\end{array} 
\qquad\quad  
\begin{array}{ccc}
\begin{tikzpicture}
    \node[MyNode] (n2) at (1.5,0) {$n_2$};
    \draw[rounded corners=3mm] (1.1,-0.3) rectangle (1.9,0.3); 
\end{tikzpicture}
& \raisebox{2ex}{$\to$} & 
\begin{tikzpicture}
    \node[MyNode] (n1) at (0,0) {$n_1$};
    \node[MyNode] (n2) at (1.5,0) {$n_2$};
    \draw[edge] (n1) to (n2);
    \node[MyEdge] at (0.75,0) {$\!e_1\!$};
    \draw[rounded corners=3mm] (-0.4,-0.3) rectangle (1.9,0.3); 
\end{tikzpicture} \\
\downarrow & & \downarrow \\ 
\begin{tikzpicture}
    \node[MyNode] (n2) at (1.5,0) {$n_2$};
    \node[MyNode] (n3) at (2.7,0) {$n_3$};
    \draw[edge] (n2) to (n3);
    \node[MyEdge] at (2.05,0) {$\!e_2\!$};
    \draw[rounded corners=3mm] (1.1,-0.3) rectangle (3.1,0.3); 
\end{tikzpicture} 
& \raisebox{2ex}{$\to$} & 
\begin{tikzpicture}
    \node[MyNode] (n1) at (0,0) {$n_1$};
    \node[MyNode] (n2) at (1.2,0) {$n_2$};
    \node[MyNode] (n3) at (2.4,0) {$n_3$};
    \draw[edge] (n1) to (n2);
    \node[MyEdge] at (0.55,0) {$\!e_1\!$};
    \draw[edge] (n2) to (n3);
    \node[MyEdge] at (1.75,0) {$\!e_2\!$};
    \draw[rounded corners=3mm] (-0.4,-0.3) rectangle (2.8,0.3); 
\end{tikzpicture} \\
\end{array} 
$$
\normalsize 
The left pushout, representing the union in $\Mod(\Sk_\cqu)$, is
clearly not pointwise, since the edge $[e_1\cric e_2]$ at the vertex
of the pushout has no antecedent. The right pushout in
$\Mod(\Rk_\cqu)$ is likewise not pointwise, despite
appearances. Recall that each object in this pushout interprets the points
$N,N',E,K$ and $K'$ of the sketch $\Rk_\cqu$. In particular, $K$ is
interpreted at the vertex of the pushout as $K =\{(e_1,e_2)\}$
representing the pair of consecutive edges $(e_1, e_2)$, whereas $K$ is
interpreted as the empty set at every object in the base of the pushout.

\end{example}

\subsec{Stuttering sketches} 
\label{ssec:stutter-sketches}

\nota 
A cone $C$ on a diagram $D$ with vertex $V$ and 
projections $p_X:V\to X$ is denoted $\coneto{p}{V}{D}$. 
The definition of a relation over a diagram $D$ 
(in Section~\ref{ssec:structure-sketches}) 
is based on two consecutive limits: 
first the limit of $D$, then a mono over the vertex of this limit. 
In this Section we show that it can be defined with only 
one limit (Theorem~\ref{theo:stutter}). 

The following Definitions~\ref{def:stt-diagram} and~\ref{def:stt-cone} for 
\emph{stuttering diagrams} and \emph{stuttering cones}, as well as 
Theorem~\ref{theo:stutter}, hold in any category with the required limits. 

\begin{definition}[stuttering diagrams]\ 
\label{def:stt-diagram}
Given a commutative cone $C$ with base $D$, 
the \emph{stuttering diagram} associated to $C$, denoted $C+_DC$, 
is made of two copies of $C$ glued along~$D$. 
\end{definition}

Let $C=\coneto{p}{V}{D}$ with $D$ denoted $X\dots Y$, then:

\footnotesize
$$  
C+_DC = 
\xymatrix@C=1pc@R=1pc{
V \ar[d]_{p_X} \ar[drr]^(.3){p_Y} && 
  V \ar[dll]_(.3){p_X} \ar[d]^{p_Y} \\ 
X & \dots & Y \\
} 
$$
\normalsize

\begin{remark}
\label{rem:stutter}
Any commutative cone $C'$ on $C+_DC$ is characterized by its vertex 
$W$ and the projections $q_1,q_2:W\to V$ for both copies of $V$, 
which must satisfy $p_X\circ q_1 = p_X\circ q_2$ for each $X$ in $D$. 
Then for each point $X$ in $D$ the projection $q_X:W\to X$ in $C'$ 
is $q_X = p_X\circ q_1 = p_X\circ q_2$. 

\footnotesize
$$  C' = 
\xymatrix@C=1.5pc@R=1pc{
& W \ar[dl]_{q_1} \ar[dr]^{q_2} & \\ 
V \ar[d]_{p_X} \ar[drr]^(.3){p_Y} && 
  V \ar[dll]_(.3){p_X} \ar[d]^{p_Y} \\ 
X & \dots & Y \\
} 
$$
\normalsize
\end{remark}

\begin{definition}[stuttering cones]\ 
\label{def:stt-cone}
Given a commutative cone $C=(\coneto{p}{V}{D})$ on $D$,  
the \emph{stuttering cone associated to} $C$, denoted $\Stt(C)$, 
is the commutative cone with base $C+_DC$, with vertex $V$ 
and with projections $\id_V:V\to V$ for both copies of $V$,  
so that its other projections are $p_X:V\to X$, as for $C$. 
A \emph{stuttering cone on} $D$ is a stuttering cone associated to
some commutative cone on $D$. 
\end{definition}

\footnotesize
$$ 
\Stt(C)= 
\xymatrix@C=1pc@R=1pc{  
& V \ar[dl]_{\id_V} \ar[dr]^{\id_V} 
  & \\
V \ar[d]_{p_X} \ar[drr]^(.3){p_Y} && 
  V \ar[dll]_(.3){p_X} \ar[d]^{p_Y} \\ 
X & \dots & Y \\
} 
$$
\normalsize

In the category of sets, 
Theorem~\ref{theo:stutter} implies that a relation, usually defined as  
a subset of a cartesian product (which involves two nested limits), 
can also be defined as the vertex of a unique limit.  

\begin{theorem}
\label{theo:stutter}
Let $D$ be a diagram, $C_0=(\coneto{p_0}{V_0}{D})$ the limit of $D$
and $C=(\coneto{p}{V}{D})$ any commutative cone on $D$.
Let $v:V\to V_0$ be the unique arrow such that
$p_X=p_{0,X}\circ v:V\to X$ for each point $X$ of $D$.
Then $v$ is a mono if and only if $\Stt(C)$ is the limit of $C+_DC$.
\end{theorem}

\begin{proof}
Using Remark~\ref{rem:stutter}, $\Stt(C)$ is the limit of $C+_DC$ 
if and only if for every cone on the diagram $C+_DC$,
with vertex $W$ and projections $q_1,q_2:W\to V$, there is 
a unique morphism $q:W\to V$ such that $q_1=q$ and $q_2=q$.
Equivalently, if and only if for every 
$q_1,q_2:W\to V$ such that $p_X\circ q_1 = p_X\circ q_2$ for each $X$ 
in $D$ we have $q_1=q_2$.
The condition $p_X\circ q_1 = p_X\circ q_2$ can be written as 
$p_{0,X}\circ v\circ q_1 = p_{0,X}\circ v\circ q_2$.
Since $C_0$ is the limit of $D$ this condition 
is satisfied for every $X$ in $D$ if and only if $v\circ q_1 = v\circ q_2$. 
Thus, $\Stt(C)$ is the limit of $C+_DC$ if and only if
for every $q_1,q_2:W\to V$ such that $v\circ q_1 = v\circ q_2$
we have $q_1=q_2$: this means that $v$ is monic.
 
$$  
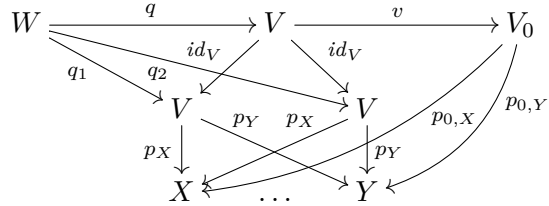
\xymatrix@C=1.5pc@R=1.5pc{
W \ar[rrd]_(.4){q_1} \ar[rrrrd]_(.4){q_2} \ar[rrr]^{q} &&  
& V \ar[dl]_{\id_V} \ar[dr]^{\id_V} \ar[rrr]^{v}  
  &&& V_0 \ar@/^3ex/[ddllll]^(.3){p_{0,X}} \ar@/^3ex/[ddll]^(.3){p_{0,Y}} \\ 
&& V \ar[d]_{p_X} \ar[drr]^(.3){p_Y} && 
  V \ar[dll]_(.3){p_X} \ar[d]^{p_Y} && \\ 
&& X & \dots & Y && \\
} 
$$
\end{proof}

\begin{definition}[stuttering sketch]\  
\label{def:stt-sketch}
A \emph{stuttering sketch} is a sketch such that all 
its potential limits are stuttering cones. 
\end{definition}

For instance, for quivers with partial concatenation,  
the stuttering sketch $\Rk'_\cqu$ below is equivalent to $\Rk_\cqu$,
as defined in Section~\ref{ssec:inference-paths}.

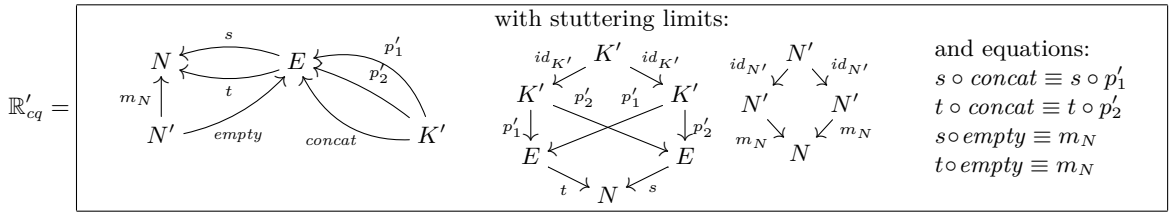
\footnotesize
$$ \Rk'_\cqu = 
\sketch{
\begin{array}{lll}
\raisebox{5ex}{
\xymatrix@C=3pc@R=1.3pc{
N & E \ar@/_/[l]_{s} \ar@/^/[l]^{t} & \\ 
N' \ar@/_/[ru]_(.4){\empt} \ar[u]^{m_N} 
&& 
K' \ar@/_4ex/[lu]_{p'_1} \ar@/_1ex/[lu]_{p'_2} 
  \ar@/^3ex/[lu]^{\concat} \\ 
} } & 
\begin{array}{l}
\mbox{with stuttering limits:} \\ 
\xymatrix@C=1pc@R=0.2pc{  
& K' \ar[dl]_{\id_{K'}} \ar[dr]^{\id_{K'}} 
  & \\
K' \ar[dd]_{p'_1} \ar[ddrr]^(.3){p'_2} && 
  K' \ar[ddll]_(.3){p'_1} \ar[dd]^{p'_2} \\ 
&& \\ 
E \ar[dr]_{t}  & & E \ar[dl]^{s} \\
& N & \\ 
} 
\xymatrix@C=0pc@R=0.5pc{  
& N' \ar[dl]_{\id_{N'}} \ar[dr]^{\id_{N'}} & \\
N' \ar[dr]_{m_N} && N' \ar[dl]^(.3){m_N} \\ 
& N & \\ 
} 
\end{array} & 
\begin{array}{l}
\mbox{and equations:} \\ 
s\circ \concat \equiv s\circ p'_1 \\ 
t\circ \concat \equiv t\circ p'_2 \\
s\scirc\empt \equiv m_N \\
t\scirc\empt \equiv m_N \\
\end{array} \\ 
\end{array}
\\ } 
$$ 
\normalsize 

\begin{theorem}
\label{theo:union} 
The finite union of compatible models of a stuttering sketch is a  
pointwise colimit. 
\end{theorem}

\begin{proof}
It is sufficient to prove this result for two models. 
In the category $\Mod(\Sk)$, the intersection $M_1\cap M_2$ is pointwise 
and the union $M_1\cup M_2$ is the colimit of the inclusions
$M_1\supseteq M_1\cap M_2 \subseteq M_2$, which is not pointwise in general.  
Let $\Qu$ be the underlying quiver of $\Sk$ and $\sk:\Qu\to\Sk$ the inclusion. Then 
the intersection of $M_1$ and $M_2$ is the same in $\Mod(\Sk)$ and in $\Mod(\Qu)$, 
and the colimit $M_3$ of $M_1\supseteq M_1\cap M_2 \subseteq M_2$ in $\Mod(\Sk)$
is the image by $F_{\sk}$ of the colimit $N_3$ of 
$M_1\supseteq M_1\cap M_2 \subseteq M_2$ in $\Mod(\Qu)$. 
Thus $N_3=M_1\cup M_2$ in $\Mod(\Qu)$, and if $N_3$ is in $\Mod(\Sk)$ 
then $M_3=N_3=M_1\cup M_2$ in $\Mod(\Sk)$.
Now let us prove that $N_3$ is in $\Mod(\Sk)$, which means that 
it takes every limit of $\Sk$ to a limit in the category of sets.
We consider any potential limit of $\Sk$, it is a cone $C'=\Stt(C)$ 
with base $C+_DC$ for some diagram $D$ and some cone $C$ with base $D$. 
Let $V$ be the common vertex of $C$ and $C'$.
Let $\Sk'$ be obtained by adding to $\Sk$ a limit of $D$, 
with vertex denoted $V_0$, and the canonical arrow $v:V\to V_0$
with the corresponding equations. Then $\Sk'$ is equivalent to $\Sk$. 
By Theorem~\ref{theo:stutter} we have to prove that $N_3(v)$ is injective. 
We have $N_3(V) = M_1(V) \cup M_2(V)$
with $M_1(V) \subseteq M_1(V_0)$ and $M_2(V) \subseteq M_2(V_0)$,
so that $N_3(V) \subseteq M_1(V_0)\cup M_2(V_0)$.
Since $M_1(V_0)\cup M_2(V_0) = N_3(V_0)$ we get $N_3(V) \subseteq N_3(V_0)$.
Thus, $N_3$ is a model of $\Sk$, so that $M_3=N_3$ 
and $M_3$ is a pointwise colimit in $\Mod(\Sk)$, as required.
\end{proof}

\begin{example}
\label{ex:stutter-pointwise} 
In the right square of Example~\ref{ex:stutter-not-pointwise}, the
union computed in $\Mod(\Rk_\cqu)$ is not pointwise, because of the
interpretation of the point $K$ from $\Rk_\cqu$.  However, as $\Rk'_\cqu$ is
equivalent to $\Rk_\cqu$, this right square is also a union in
$\Mod(\Rk'_\cqu)$ and this union is pointwise.
\end{example}

\section{Conclusion} 
\label{sec:conclusion}
In this paper, we demonstrated how sketch-oriented databases, built on
the foundation of finite-limit sketches, provide a powerful and
flexible framework for modeling graph databases. By leveraging
category theory, we have shown that finite-limit sketches can
effectively represent various database paradigms, including RDF
graphs and property graphs. Application to
relational databases can be easily obtained by encoding relations and
tables (see Appendix~\ref{app:tables}).
The diagrammatic nature of sketches is close to established database
formalisms, such as Entity-Relationship diagrams, UML and graphical
schema languages.
Advanced features such as inference
systems, paths, schemas, and typing were explored, highlighting their
compatibility with sketch-oriented approaches. The introduction of
stuttering sketches offers a novel method for simplifying database
structures and ensuring pointwise unions of models, addressing
challenges in managing complex data relationships. This concept opens
new avenues for specifying database operations and maintaining
consistency across databases.  Sketch-oriented databases unify
theoretical rigor with practical applicability, making them a
promising tool for representing and querying complex data
structures. Future work include extensions of the presented framework
to algebraic query answering, as well as the investigation of additional
applications in semantic web and ontology-based systems.


\newpage
\appendix
\section{Relational databases}\
\label{app:tables}   
Although this paper focuses on graph databases, it is worth noting that  
the relational database paradigm can also be seen as sketch-oriented.  
In this Appendix we propose some hints for this purpose. 
More issues, like adding constraints about primary keys and 
foreign keys relating several tables, could be added. 

A \emph{table} is made of a set of \emph{rows}, 
a set of \emph{columns} and a set of \emph{entries}, 
with at most one entry $e(x,y)$ for each row $x$ and column $y$. 
Thus, a table is a model of the following sketch, 
with $\Roww$, $\Coll$, $\Ent$, $\Cell$ and ${\it NECells}$ 
interpreted respectively as the sets of rows, columns, entities, 
cells and non-empty cells. 

\small
$$ \Sk_\tabl = 
\sketch{
\raisebox{4ex}{$
\xymatrix@C=1.5pc@R=.3pc{
&& \Cell \ar[lld]_{\row} \ar[rd]^{\col} 
  & \\ 
\Roww & \Ent\! & & \Coll \\ 
&& \NECell \ar[uu]_(.4){m_{\Cell}} \ar[lu]^(.6){\ent} & \\ 
}
$}  
\quad 
\begin{array}{l}
\mbox{with a product:} \\
\xymatrix@C=-.5pc@R=1pc{
& \Cell \ar[dl]_{\row} \ar[dr]^{\col} & \\ 
\Roww && \Coll \\ 
} 
\\
\end{array} 
\quad 
\begin{array}{l}
\mbox{and a mono:} \\ 
\xymatrix@C=2pc{
\NECell\, \ar@{>->}[r]^(.55){m_{\Cell}} & \Cell \\ 
} 
\\
\end{array} 
\\ } 
$$ 
\normalsize

Besides, one can observe that each table can be associated to a quiver
with attribute-value pairs on nodes, as follows: no edge, one node for
each row of the table, one attribute for each column, one value for
each entry, and for each non-empty cell $(r,c)$ one attribute-value
pair $c\prttt v$ on the node $r$, where $v= {\it entry}(r,c)$.  Thus,
as in Section~\ref{sssec:structure-avpairs}, each table can be
associated to a strongly labeled quiver, with one node for each row of
the table plus one node for each entry, one label for each column, and
for each non-empty cell $(r,c)$ one edge $\tr{(r)}{c}{(v)}$, where
$v= {\it entry}(r,c)$.
The fact that each table is associated to a quiver with attribute-value pairs
on nodes can be seen as the application of the underlying functor associated 
to the pleomorphism from $\Sk_{\avp,X}$ to $\Sk_\tabl$ which takes 
$X,V,A$ to $\Roww,\Ent,\Coll$, thus $T'$ to the product 
$\Roww\stimes\Ent\stimes\Coll$, isomorphic to $\Cell\stimes\Ent$, 
and which takes $T$ to $\NECell$ and $m_T$ to the pair made of 
$m_\Cell$ and $\ent$. 

\begin{example} 
\label{ex:tables} 
Part of the data in Example~\ref{ex:structure-schema} is now stored as a table: 

\small
$$ 
\begin{array}{|l|l|l|l|}   
\hline
 & \Pe & \Bo & \Jo \\ 
\hline
r_1 & {\rm MB} & {\rm CTCS} & {\rm RTAC} \\ 
r_2 & {\rm CW} & {\rm CTCS} & {\rm RTAC} \\ 
r_3 & {\rm SML} & {\rm CWM} & \\ 
\hline
\end{array}
$$
\normalsize
Here is the associated quiver with attribute-value pairs: 

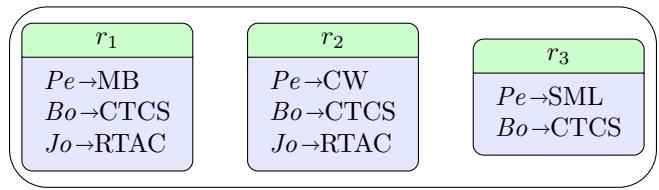
\small
$$
\begin{tikzpicture} 
    \node[MyNodeDouble] (r1) at (0,0) {
      $r_1$
      \nodepart{two}
        {\small $ \begin{array}{l} 
        \Pe\prttt{\rm MB} \\ 
        \Bo\prttt{\rm CTCS} \\ 
        \Jo\prttt{\rm RTAC} \\ 
        \end{array} $ }};
    \node[MyNodeDouble] (r2) at (3,0) {
      $r_2$
      \nodepart{two}
        {\small $ \begin{array}{l} 
        \Pe\prttt{\rm CW} \\ 
        \Bo\prttt{\rm CTCS} \\ 
        \Jo\prttt{\rm RTAC} \\ 
        \end{array} $ }};
    \node[MyNodeDouble] (r3) at (6,0) {
      $r_3$
      \nodepart{two}
        {\small $ \begin{array}{l} 
        \Pe\prttt{\rm SML} \\ 
        \Bo\prttt{\rm CTCS} \\ 
        \end{array} $ }};
    \draw[rounded corners=5mm] (-1.3,-1.2) rectangle (7.3,1.2); 
\end{tikzpicture}
$$
\normalsize
and the associated strongly labeled quiver
(for readability, only some labels are noted):

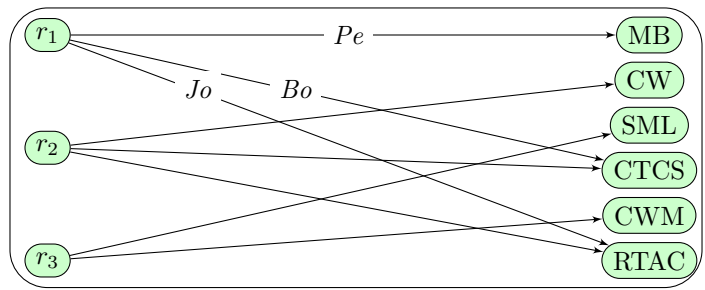
\small
$$  
\begin{tikzpicture}[yscale=0.6] 
    \node[MyNode] (r1) at (0,5) {$r_1$};
    \node[MyNode] (r2) at (0,2.5) {$r_2$};
    \node[MyNode] (r3) at (0,0) {$r_3$};
    \node[MyNode] (mb) at (8,5) {MB};
    \node[MyNode] (cw) at (8,4) {CW};
    \node[MyNode] (sml) at (8,3) {SML};
    \node[MyNode] (ctcs) at (8,2) {CTCS}; 
    \node[MyNode] (cwm) at (8,1) {CWM};
    \node[MyNode] (rtac) at (8,0) {RTAC};
    \draw[edge] (r1) to (mb);
    \draw[edge] (r1) to (ctcs);
    \draw[edge] (r1) to (rtac);
    \draw[edge] (r2) to (cw);
    \draw[edge] (r2) to (ctcs);
    \draw[edge] (r2) to (rtac);
    \draw[edge] (r3) to (sml);
    \draw[edge] (r3) to (cwm);
    \node[MyAtt] (pe1) at (4,5) {$\Pe$};
    \node[MyAtt] (bo1) at (3.3,3.8) {$\Bo$};
    \node[MyAtt] (jo1) at (2,3.8) {$\Jo$};
    \draw[rounded corners=5mm] (-0.5,-0.6) rectangle (8.7,5.6); 
\end{tikzpicture}
$$
\normalsize
\end{example} 

This point of view is compatible with the notion of schema and typing 
in Section~\ref{sssec:structure-schema}.

\begin{example} 
\label{ex:tables-schema} 
The schema for Example~\ref{ex:tables} is the one-row table: 

\small
$$ 
\begin{array}{|l|l|l|l|}   
\hline
 & \Pe & \Bo & \Jo \\ 
\hline
{\it Biblio} & \String & \String & \String \\ 
\hline
\end{array}
$$
\normalsize
Here are the associated quiver with attribute-value pairs and 
the associated strongly labeled quiver:

\small
$$
\begin{tikzpicture} 
    \node[MyNodeDouble] (r1) at (0,0) {
      {\it Biblio} 
      \nodepart{two}
        {\small $ \begin{array}{l} 
        \Pe\prttt\String \\ 
        \Bo\prttt\String \\ 
        \Jo\prttt\String \\ 
        \end{array} $ }};
    \draw[rounded corners=5mm] (-1.3,-1.2) rectangle (1.3,1.2); 
\end{tikzpicture}
\qquad \qquad 
\begin{tikzpicture} 
    \node[MyNode] (bib) at (0,0) {${\it Biblio}$};
    \node[MyNode] (str) at (4,0) {$\String$};
    \draw[edge][bend left=20] (bib) to (str);
    \draw[edge] (bib) to (str);
    \draw[edge][bend right=20] (bib) to (str);
    \node[MyAtt] (pe) at (2,0.5) {$\Pe$};
    \node[MyAtt] (bo) at (2,0) {$\Bo$};
    \node[MyAtt] (jo) at (2,-0.5) {$\Jo$};
    \draw[rounded corners=5mm] (-0.7,-0.7) rectangle (4.7,0.7); 
\end{tikzpicture}
$$
\normalsize
\end{example}

\end{document}